\documentclass[runningheads]{llncs}
\usepackage{graphicx}
\usepackage{amsmath}
\usepackage{amssymb}
\usepackage{enumitem}
\usepackage{tikz}
\usepackage{subcaption}

\usepackage{algorithm}
\usepackage{bbm}

\usepackage{listings}
\usepackage{natbib}
\setcitestyle{authoryear,open={[},close={]}}

\newcommand{\R}{\mathbb{R}}
\newcommand{\N}{\mathbb{N}}
\newcommand{\Z}{\mathbb{Z}}

\newcommand{\abs}[1]{\lvert #1 \rvert}
\newcommand{\parenthesis}[1]{\left( #1 \right)}
\newcommand{\brackets}[1]{\left[ #1 \right]}
\newcommand{\braces}[1]{\{ #1 \}}

\newcommand{\cS}{\mathcal{S}}
\DeclareMathOperator{\argmin}{argmin}

\newcommand{\dyn}[1]{\Tilde{#1}}
\newcommand{\better}[2]{\mathcal{B}_{#1}^{{#2}}}
\newcommand{\pwin}[1]{j^*\!\parenthesis{#1}}
\newcommand{\swin}[1]{\wm{s^{#1}}}
\newcommand{\swm}[1]{\twm^{#1}}
\newcommand{\pmov}[1]{j^{#1}}
\newcommand{\smov}[1]{s'\!\parenthesis{#1}}
\newcommand{\tmed}{\mathrm{med}}
\newcommand{\twm}{\mathrm{wm}}
\newcommand{\statemed}[1]{\mathrm{med}\!\parenthesis{#1}}

\newcommand{\med}[1]{\mathrm{med}^{#1}}
\newcommand{\wm}[1]{\mathrm{wm}\!\parenthesis{#1}}

\newcommand{\stam}[1]{}

\newcommand{\IS}[2]{P^\sigma(#1,#2)}
\newcommand{\proxyradius}{3}
\newcommand{\voterradius}{1}

\newtheorem{observation}{Observation}

\begin{document}
	\title{Strategic Proxy Voting on the Line}
	%
	%
	\author{Gili Bielous\inst{1} \and Reshef Meir\inst{2}}
	%
	\authorrunning{G. Bielous \and R. Meir}
	%
	\institute{Technion- Israel Institute of Technology, Haifa, Israel
		\email{gili.bielous@campus.technion.ac.il} \and
		\email{reshefm@ie.technion.ac.il}}
	\maketitle              
	\begin{abstract}
		This paper offers a framework for the study of strategic behavior in proxy voting, where non-active voters delegate their votes to active voters. We further study how proxy voting affects the strategic behavior of non-active voters and proxies (active voters) under complete and partial information. We focus on the median voting rule for single-peaked preferences. 
		
		Our results show strategyproofness with respect to non-active voters. Furthermore, while strategyproofness does not extend to proxies, we show that the outcome is bounded and, under mild restrictions, strategic behavior  leads to socially optimal outcomes. 
		
		We further show that our results extend to partial information settings, and in particular for regret-averse agents.
		
		\keywords{Computational Social Choice \and Proxy Voting \and Strategic Voting \and
			Strategyproofness.}
	\end{abstract}
	\section{Introduction}
	
	In the age of internet, we see an increase of platforms and mechanisms for collective decision-making. However, many of these platforms suffer from low participation rates~\citep{schaupp2005voting,jonsson2011user}. Thus, while there is an increase in the ability of individuals to influence collective decision-making in many areas, most decisions are made by a small, non-elected and non-representative groups of active voters. Partial participation may increase vote distortion~\citep{ghodsi2019distortion} (the worst-case ratio between the social cost of the candidate elected and the optimal candidate, first defined in~\citep{procaccia2006distortion}); 
	lead to counter-intuitive equilibria~\citep{desmedt2010equilibria}; and significantly decrease the likelihood of selecting the Condorcet winner (when it exists)~\citep{gehrlein2010voting}. Above all, when the outcome of an election only considers a fraction of all opinions, it is unreasonable to assume that they accurately reflect the aggregated opinions of the collective.
	
	Proxy voting, a long standing practice in politics and corporates~\citep{riddick1991riddick}, and an up-and-coming practice in e-voting and participatory democracies~\citep{petrik2009participation}, aims at mitigating the adverse effects of partial participation. Non-active voters (followers) delegate their vote to another active voter (proxy), thereby at least having some influence on the outcome. \citet{cohensius2016proxy} proposed a model where the voters are sampled from a given distribution of non-atomic voters. Among them are a subset of proxies, with voting power proportional to the population mass that delegates to them. The outcomes of various voting rules, in particular the median voting rule, as determined by the voters are compared against the outcome via proxy voting. They show that for most settings, the outcome via proxy voting improves the accuracy with respect to the aggregated social preference of the entire population.
	
	However, such delegation changes the power dynamic of voters by shifting some of the voting power to proxies. While much consideration is granted in the literature of social choice for the strategic behavior of voters~\citep{gibbard1973manipulation,satterthwaite1975strategy} and candidates~\citep{dutta2001strategic,sabato2017real}, there is little consideration of the \emph{strategic behavior of proxies or followers} in proxy-mediated settings. ~\citet{cohensius2016proxy} consider strategic participation (i.e. selecting to participate or abstain) with mostly positive results. Notably, they show convergence to an equilibrium with the same accuracy as without strategic behavior using proxy voting. Yet they pose the question of strategic behavior of followers as an open question, which was part of the inspiration to the current study.
	
	Moreover, it is common to study strategic behavior in adversarial settings assuming complete information. This makes sense as a worst-case assumption for strategyproofness, but treating uncertainty is unavoidable when we are trying to model actual strategic voting and predict its implications (for an overview of uncertainty in voting and equilibrium models, see~\citet{meir2018strategic} Chapters 6 and 8). 
	In the context of proxy voting, assuming full information is even less reasonable: by delegating their vote, followers may wish to avoid the cognitive strain, time loss and other costs associated with determining and communicating their position. Thus, a setting that requires followers to explicitly define their positions negates these benefits of proxy voting for followers. 
	
	Therefore, it makes more sense that active voters set their strategies based on partial information on the positions of potential followers.  While there are many ways to model such uncertainty, we adopt the framework of ~\citet{reijngoud2012voter} that allows for simple and flexible definition of information sets.
	
	\subsection{Related Work}
	
	The effects of delegation on the accuracy of results have been recently studied in the context of \emph{liquid democracy}, a delegation model where voters may continue to transitively delegate their votes. \citet{kahng2021liquid,caragiannis2019contribution} show that the concentration of power in liquid democracy can be so severe that it leads to low accuracy with respect to an assumed ground truth. Subsequent work attempted to limit power concentration, either by an impartial planner~\citep{golz2021fluid} or in by altering the delegation mechanism~\cite{halpern2021defense}. In contrast, as previously mentioned \citet{cohensius2016proxy} achieved positive results for one-step (proxy) delegation.
	
	There are two closely related spatial models to our setting. The first is the model of ~\citet{cohensius2016proxy} mentioned above. The second is \emph{Strategic Candidacy Games} proposed by ~\citet{sabato2017real}, where candidates are assumed to have self-supporting preferences over possible outcomes. Thus, candidates have incentives to strategize even when they cannot guarantee their own win. Their results show the existence of a Nash equilibrium for Condorcet-consistent voting rules for voters with symmetric single-peaked preferences and every set of preferences for candidates. Our model differs and generalizes in the following sense. First, in our model, candidates' preferences are based on the outcome, not the identities of candidates. In particular, the preferences are determined \textit{ex-post} for a given state. Second, candidates in \cite{sabato2017real} are weightless, while voters (equivalent to followers in our model) are atomic. Our results demonstrate a stronger claim than the mere existence of a Nash Equilibrium, even for this generalized model. In particular, we show convergence to NE similar to the one described by their work (subject to certain restrictions). As a result, our work significantly expands upon theirs. Moreover, some of our results become trivialized when proxies are non-atomic.
	
	\subsection{Contribution and Paper Structure}
	The use of spatial models for the study of behavior and results of voting mechanisms have been first introduced by ~\citet{hotelling1929stability} and~\citet{downs1957economic}. Our model follows this approach, 
	is too a spatial model, assuming the political spectrum is represented as positions on the real line. We focus on the median voting rule that has been shown to be (group) strategyproof for single-peaked preferences~\citep{black1948rationale,moulin1980strategy}.
	
	It is important to stress that the objectives of candidates in Hotelling-Downs is different than our setting. Proxies want to maximize the outcome with respect to their preferences whereas in HD candidates wish to maximize their votes. While vote maximizing seems like a winning strategy in this context, in fact the winning strategy is to restructure the partition of votes. We show this in some section.
	
	Our initial study considers strategyproofness and manipulability with respect to both followers and proxies positions. Then, we consider sequences where proxies react to other proxies' actions. Finally, we turn to study strategic behavior in partial information settings.
	
	Our contribution is as follows:
	\begin{itemize}[topsep=-1mm,itemsep=-4pt,partopsep=0pt,parsep=1ex]
		\item Followers never have an incentive to misreport their position (or, equivalently, to follow a proxy other than the nearest one).[Section~\ref{sec:proxy_manipulations}]
		\item Proxy voting with the median voting rule is \emph{manipulable} with respect to proxy positions, and we provide a complete characterization of manipulable scenarios.[Section~\ref{sec:proxy_manipulations}]
		\item In sequences of manipulations, the outcome of each step is bounded.[Section~\ref{sec:proxy_manipulation_better_outcomes}]
		\item Under mild restrictions, sequences of manipulations converge to an optimal equilibrium.[Section~\ref{sec:proxy_manipulation_better_outcomes}]
		\item Manipulations under partial information may converge to a worse equilibrium than without delegation.[Section~\ref{sec:partial_information}]
		\item If agents are regret-averse, then manipulations converge to a socially optimal equilibrium even with partial information.[Section~\ref{sec:partial_information}]
	\end{itemize}
	
	A preliminary version of this paper was presented in EUMAS2022~\citep{bielous2022proxy}. In this version, we offer two significant additions to our results. First, we generalize our positive results for unrestricted manipulations by providing a bound on the outcome in strategic proxy settings.
	Second, in this version we show that our results extend beyond complete information settings to partial information scenarios. We show that even with limited information, the proxy voting framework retains its desirable properties. In particular, we examine the implications for regret-averse agents and find that our results hold in these cases as well.
	
	By investigating the effects of strategic behavior in proxy voting under different information conditions, our study contributes to a deeper understanding of the dynamics and potential benefits of this voting mechanism.
	
	
	\section{Model and Preliminaries}
	
	We define the model of \emph{Strategic Proxy Games (SPG)} as follows.
	
	\paragraph{Model.}
	Our basic model follows the one in by \citet{cohensius2016proxy}. There is a set of proxies (active agents) $M = \braces{1,...,m}$, and a set of followers $N = \braces{1,...,n}$. We refer to the set of all agents $N \cup M$ as 'voters'. Each voter $1 \le i \le n+m$ has a position $p_i \in \R$ along the political spectrum. True positions are $p \in \R^{m+n}$, where $p|_M := (p_j)_{j\in M}$ and $p|_N := (p_i)_{i\in N}$.
	A \emph{state} is a vector $s \in \R^m$, such that $s_j$ is the position Proxy $j$ declares. We denote by $\parenthesis{s_{-j},s_j'}$ the state that is equal to $s$ except for the strategy of Proxy $j$, that is $s_j'$. 
	
	\paragraph{Delegation.}
	We assume that each follower delegates their vote to the nearest proxy (this is known as the Tullock delegation model~\citep{tullock1967proportional}). Formally, given a vector of positions $p$ and a state $s$, each Follower $i \in N$ delegates their vote to Proxy $j \in M$, where $$\varphi_i\parenthesis{s} := \argmin_{j \in M} \abs{s_j - p_i}$$
	We assume the existence of a deterministic tie-breaking scheme that only depends on the state of voters.
	All proxies delegate their vote to themselves.
	
	\paragraph{Preferences}
	Voters are assumed to have single-peaked preferences with peak at $p_i$. That is, for every $x,y \in \R$, if $x < y \le p_i$, then Voter $i$ prefers $y$ to $x$, and if $p_i \le x < y$, then Voter $i$ prefers $x$ to $y$. For followers we further assume preferences are symmetric, that is, for every $x,y \in \R$, if $\abs{x-p_i} < \abs{y-p_i}$, then Follower $i$ prefers $x$ to $y$. Thus, preferences of voters are consistent with the delegation model.
	
	\begin{example}
		\label{xmpl:perm}
		Consider the SPG appearing in Figure~\ref{fig:model_example}.
		
		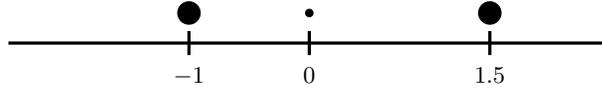
\begin{figure}[ht!]
			\centering
			\begin{tikzpicture}[
				line/.style={draw=black, very thick},
				proxy/.style={circle,draw=black,fill,inner sep=\proxyradius pt},
				voter/.style={circle,draw=black,fill,inner sep=\voterradius pt},
				x=0.8mm, y=0.8mm, z=1mm]
				
				\draw[line, -](-50,30)--(50,30);
				
				\foreach \x/\t in {-1/1,0/0,1.5/1} {
					\draw[line](20*\x,28)--(20*\x,32) node[below=10]{$\x$};
					\ifthenelse{\t=1}
					{\draw (20*\x,35) node[proxy] {};}
					{\draw (20*\x,35) node[voter] {};}
				}
			\end{tikzpicture}
			\vspace{-4mm}
			\caption{An example SPG. Large dots indicate the positions of proxies, small dots indicate the positions of followers.\vspace{-4mm}}
			\label{fig:model_example}
		\end{figure}
		There are two proxies $\braces{1,2}$ with positions $p_1=-1$ and $p_2=1.5$.  There is a single follower $\braces{3}$ with position $p_3=0$. In the truthful state $s=p|_M=(-1,1.5)$, the follower delegates their vote to $\varphi_3\parenthesis{s}=1$. Thus, there are two votes to $-1$ and a single vote to $1.5$.
	\end{example}
	
	\paragraph{Weighted median.}
	Given a finite vector $\Vec{s} \in \R^m$ such that each $s_i \in \Vec{s}$ has weight $w_{s_i} \in \R^{+}$, let $W = \sum_{s_i \in \Vec{s}} w_{s_i}$. The weighted median of $\Vec{s}$, denoted $\statemed{\Vec{s};w}$, is $s_i \in \Vec{s}$ such that
	$$\sum_{\braces{s_j \in \Vec{s} \setminus \braces{s_i}:s_j \le s_i}} w_{s_j} \le \frac{W}{2}\quad \mathrm{and}\quad \sum_{\braces{s_j \in \Vec{s} \setminus \braces{s_i}:s_j \ge s_i}} w_{s_j} \le \frac{W}{2}.$$
	That is, the sum of weights of elements that are smaller than $s_i$ is at most half the total sum of weights, and the same holds for the sum of weights of elements that are larger than $s_i$.
	
	\paragraph{Weighted median voting rule.}
	Next, we define the Weighted Median voting rule. The weight of each proxy is defined as the number of delegations to them. Then, the \emph{weighted median voting rule} (WM) selects the position that is the weighted median of proxy positions.
	Formally:
	$$\statemed{s,p|_N} := \statemed{(s,p|_N) ; 1} $$
	is the unweighted median of all voters (proxies and followers) at state $s$, and the weighted median is
	$$\wm{s,p|_N} := \statemed{s; w} \text{ where } w_j := \abs{\braces{i\in N: \varphi_i(s)=j}} +1,$$
	and we often omit $p|_N$ when clear from context. Ties break lexicographically.
	
	At state $s$, the WM voting rule selects $\wm{s,p|_N} \in \R$ as the winner. Note that $\wm{s}=s_j$ for some $j\in M$, and we denote this selected proxy  by $\pwin{s,p|_N} \in M$.
	We denote the median and weighted median in the true state $p$ by  $\tmed:=\statemed{p}$ and $\twm:=\wm{p}$, respectively.
	
	The WM winner in Example~\ref{xmpl:perm} is the position $-1$. This is because the proxy at $-1$ receives a total of $2$ votes from the single follower who delegates to them, whereas the proxy at $1.5$ receives $1$ vote.
	
	\paragraph{Strategyproofness and Manipulations.}
	We say that a voter is \emph{truthful} if they declare their true position $p_i$. Voters may lie about their positions, and we assume that voters are rational, that is, they lie only if by lying the outcome changes in their favor. 
	
	We say that $p'_i\neq p_i$ is a \emph{manipulation} for voter $i \in N \cup M$  if voter~$i$ strictly prefers $\wm{p'}$ to $\wm{p}$, where  $p'=\parenthesis{p_{-i},p_i'}$.
	
	A voting rule is \emph{strategyproof} if for every vector of true positions $p$, no voter has a manipulation; otherwise, it is \emph{manipulable}. The Median voting rule, i.e. the voting rule that selects the unweighted median is known to be (group) strategyproof for single-peaked preferences~\citep{black1948rationale,moulin1980strategy}, and thus in particular to voters who try to minimize their distance as in our model.
	
	
	\section{Strategyproofness of Weighted Median}
	
	\label{sec:proxy_manipulations}
	\subsection{Manipulation by Followers}
	We begin our analysis by showing that strategyproofness extends to Weighted Median with respect to followers' positions. In their work, \citet{cohensius2016proxy} demonstrate that for any distribution of followers and proxies, the WM winner is the proxy who is closest to the true median. That is, the proxy $j^*$ selected by the weighted median rule is the one closest to the (unweighted) median of the entire population. Equivalently, it is the proxy selected by the median voter in the population. Formally, let $i^*$ be the median voter in profile $(s,p|_N)$. Then:
	\begin{lemma}[\cite{cohensius2016proxy}]
		\label{lm:nearest_proxy_to_median_wins}
		$$j^*(s)= \varphi_{i^*}(s)=\argmin_{j\in M}|s_j-\statemed{s,p|_N}|.$$
	\end{lemma}
	
	Next, we prove that for WM, followers do not have manipulations. Note that for followers, manipulations are in effect delegation to another proxy.
	
	\begin{theorem}
		\label{thm:voters_strategyproofness}
		WM is strategyproof w.r.t followers' positions.
	\end{theorem}
	
	\begin{proof}
		Assume towards contradiction that for some $p$, there exists a follower $i \in N$ who has a manipulation $p_i'$, and denote $p'=(p_{-i},p'_i)$. W.l.o.g, assume $p_i \ge \statemed{p}$. 
		
		By Lemma~\ref{lm:nearest_proxy_to_median_wins} the weighted median rule returns the proxy $j\in M$ closest to $\statemed{p}$ before the manipulation, and $j'$ (who is closest to $\statemed{p'}$) after the manipulation. 
		
		Accordingly, the winning positions before and after are $p_j = \wm{p|_M,p|_N}$ and $p_{j'} = \wm{p|_M,p'|_N}$.
		
		As $p_i'$ is a manipulation, we have that $p_{j'}$ is strictly closer to $p_i$ than $p_j$, and in particular $j\neq j'$. This must mean that $\statemed{p'}\neq \statemed{p}$, and by strategyproofness of the unweighted median, $\statemed{p'}<\statemed{p}<p_i$.
		
		By monotonicity of the unweighted median this means $p'_i<p_i$, and thus $p_{j'}<p_j$.
		One of the following must hold:
		\begin{itemize}
			\item $p_i \leq p_{j'}< p_j$: Since $\statemed{p} \le p_i$, it follows that $j$ is closer to $\statemed{p}$ than $j'$, in contradiction to the definition of $j$.
			\item $p_{j'}<p_j \leq p_i$: This is not a manipulation for $i$, as the outcome gets farther away.
			\item $p_{j'} < p_i < p_j$: Since $p_i \ge \statemed{p}$, and since $\abs{\statemed{p}-p_j} \leq \abs{\statemed{p}-p_{j'}}$, we get $\abs{p_i-p_j} < \abs{p_i-p_{j'}}$. Thus, by symmetric single-peakedness Follower $i$ prefers $p_j$ to $p_{j'}$, in contradiction to $p_i'$ being a manipulation.\qed
		\end{itemize}
	\end{proof}
	
	\begin{remark}
		Another interpretation of Theorem~\ref{thm:voters_strategyproofness} is that under the weighted median voting rule, it is a dominant strategy for a follower to support her nearest proxy. In other words, the theorem \emph{justifies} the Tullock delegation model.
	\end{remark}
	As Theorem~\ref{thm:voters_strategyproofness} shows that WM is strategyproof with respect to the positions of followers, we can henceforth consider them as non-strategic agents. In what follows, followers are considered to always be truthful. In particular, the position vector $p|_N$ is fixed.
	
	\subsection{Manipulation by Proxies }
	We continue to analyze the strategic behavior of proxies. While we obtain a positive result of strategyproofness when only followers are strategic, the same does not hold for proxies, as demonstrated by the following example.
	
	\begin{example}
		\label{xmpl:proxy_manipulable}
		Recall the SPG appearing in Example~\ref{xmpl:perm}. The truthful WM winner is $\wm{s} = -1$, and the winning proxy is $1$. Consider the state $s=\parenthesis{p_1, 1-\varepsilon}$ for some $0 < \varepsilon < 2$.
		\begin{figure}[ht]
			\centering
			\begin{tikzpicture}[
				line/.style={draw=black, very thick}, 
				x=0.8mm, y=0.8mm, z=1mm]
				
				\draw[line, -](0,30)--(100,30);
				\draw[line](30,28)--(30,32) node[below=10]{$-1$};
				\draw (30,35) node[circle,fill,inner sep=3pt] {};
				\draw (50,35) node[circle,draw=black,fill=none,inner sep=1pt] {};
				\draw[line](50,28)--(50,32) node[below=10]{$0$};
				\draw[line](65,28)--(65,32) node[below=10]{$1-\varepsilon$};
				\draw (65,35) node[circle,fill,inner sep=3pt] {};
				\draw (65,40) node[circle,fill,inner sep=1pt] {};
				\draw[line](80,28)--(80,32) node[below=10]{$1.5$};
				\draw (80,35) node[circle,draw=black,fill=none,inner sep=3pt] {};
			\end{tikzpicture}
			\caption{The SPG with manipulation by Proxy $2$. Large empty dot is Proxy $2$'s true position, the manipulation is reporting $1-\varepsilon$. The single follower delegates their vote to Proxy $2$.}
			\label{fig:manipulation}
			\vspace{-1em}
		\end{figure}
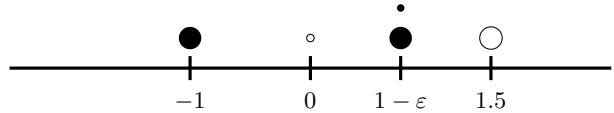
		
		The single follower delegates their vote to Proxy $2$. There are two votes to $s_2 = 1-\varepsilon$ and only one vote to $s_1 = -1$, thus, $\wm{s}=1-\varepsilon$.
		As preferences are single-peaked and Proxy $2$'s peak is at $p_2 = 1.5$, we get that Proxy $2$ strictly prefers $1-\varepsilon$ to $-1$. Hence, $1-\varepsilon$ is a manipulation for Proxy $2$.
	\end{example}
	
	The counter-example presented in Example~\ref{xmpl:proxy_manipulable} can be easily expanded to any number of followers and proxies. However, rather than formally constructing such example, The following theorem provides a complete characterization of manipulable scenarios. As a consequence, it shows that manipulations exist under very simple and reasonable conditions.
	
	\begin{theorem}
		\label{thm:manipulable_proxy_characterization}
		There is a proxy that has a manipulation in the truthful state $p|_M$ iff it holds that $p_j \ne \tmed$ for all $1 \le j \le m$, and there are proxies $j,j' \in M$ such that $p_j < \tmed < p_{j'}$.
	\end{theorem}
	\begin{proof}
		``$\Leftarrow$'' Suppose $p_j < \tmed < p_{j'}$, and w.l.o.g. let $j$ be the closest proxy to the median, so $\pwin{p}=j$, and $\twm=p_j$. 
		As preferences are single-peaked, Proxy $j'$ prefers $\tmed$ to $\wm{p}$. We proceed by showing that moving to $p'_{j'} = \tmed$ is a manipulation for Proxy~$j'$, as in Fig.~\ref{fig:f2}.
		
		Indeed, denote $p'=(p_{-j},p'_{j'})$. Since $p_j,p'_{j'}$ are on the same side of $\tmed$, we have that the position of the median does not change, i.e., $\statemed{p'}=\tmed$. 
		
		By Lemma~\ref{lm:nearest_proxy_to_median_wins}  we get that $\wm{p'}$ is  the proxy closest to $\statemed{p'}$, whose position is 
		$p'_{j'}=\tmed$.
		Since $\wm{p'}=p'_{j'}$ is strictly between $p_{j'}$ and $\twm$, this is a manipulation for $j'$.

		\begin{figure}[t]
			\centering
			\vspace{-0mm}
			\begin{tikzpicture}[
				line/.style={draw=black, very thick},
				proxy/.style={circle,draw=black,fill,inner sep=\proxyradius pt},
				voter/.style={circle,draw=black,fill,inner sep=\voterradius pt},
				x=0.8mm, y=0.8mm, z=1mm]
				
				\draw[line, -](-50,0)--(50,0);
				
				\foreach \x/\proxyind/\t in {-1/1/\twm=p_j,0/0/\tmed,1.5/1/p_{j'}} {
					\draw[line](20*\x,-2)--(20*\x,2) node[below=10]{$\t$};
					\ifthenelse{\proxyind=1}
					{\draw (20*\x,5) node[proxy] {};}
					{\draw (20*\x,5) node[voter] {};}
				}
				
				\draw[line,-latex] (20*1.5-2,7) to[out=155,in=25]  node[above] {$s_{j'}$} (0,7);
			\end{tikzpicture}
			\caption{A proxy with a manipulation.\label{fig:f2}
				\vspace{-1em}
			}
		\end{figure}
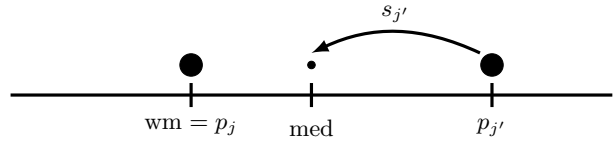
		
		``$\Rightarrow$''. If there is some proxy $k$ such that $p_{k} = \tmed$, then by Lemma~\ref{lm:nearest_proxy_to_median_wins}, $\tmed$ is the WM winner. Therefore, every proxy with position at $\tmed$ have their peak outcome, so there is no more preferred outcome for them. Consequently they do not have a manipulation.
		Further, no position is closer to $\tmed$, thus no other proxy can change the outcome by reporting a position that is closer to the median. Therefore, they can only manipulate by reporting a position that changes the position of the median. 
		
		Assume towards contradiction that there is such a proxy $k$ with a manipulation $p'_{k}$, and let $p'=\parenthesis{p_{-k},p'_{k}}$. W.l.o.g assume that $p_{k} > \tmed$. Then, the position of the median would only change in $p'$ if Proxy $k$ reports a position on the other side of $\tmed$, i.e. $p'_{k} < \tmed$. We get that $\statemed{p'}<\tmed$. Since $p'_{k}$ is a manipulation, the outcome of $p'$ holds 
		$$\wm{p'} \le \statemed{p'} < \tmed=\twm < p_{k}.$$
		The first inequality is since $\statemed{p'}$ is the maximal position $p'_i<\tmed$ for any $i\in M\cup N$. By single-peakedness $k$ prefers $\tmed$ to $\wm{p'}$, in contradiction to $p'_{k}$ being a manipulation.
		
		Finally, assume that for all proxies $p_{k} \le \twm < \tmed$. Clearly the proxy $j$ who is closest to $\tmed$ have no manipulation since their position $p_j$ wins. Also note that in any manipulation $p'_k$ we have $\statemed{p'}\geq \tmed$. The only way for $k$ to change the outcome is by becoming the winner themselves, i.e. reporting a position closer to $\statemed{p'}$ than $p_j$. However since $p_j < \tmed \leq \statemed{p'}$, we have $p_k < \twm = p_j < \wm{p'}$ and thus Proxy $k$ strictly loses from such a move.
		\qed
		
		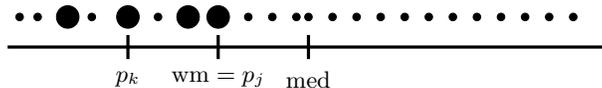
\begin{figure}[t]
			\centering
			\begin{tikzpicture}[
				line/.style={draw=black, very thick},
				proxy/.style={circle,draw=black,fill,inner sep=\proxyradius pt},
				voter/.style={circle,draw=black,fill,inner sep=\voterradius pt},
				x=0.8mm, y=0.8mm, z=1mm]
				
				\draw[line, -](-50,0)--(50,0);
				
				\draw[line](-30,-2)--(-30,2) node[below=10]{$p_k$};	
				\draw[line](-15,-2)--(-15,2) node[below=10]{$\twm=p_j$};
				\draw[line](0,-2)--(0,2) node[below=10]{$\tmed$};
				
				\foreach \x in {-4,-3,-2,-1.5} {
					\draw (10*\x,5) node[proxy] {};
				}
				
				\foreach \x in {-4.5,-4.8,-3.6,-2.5,-1,-0.6,-0.2,0} {
					\draw (10*\x,5) node[voter] {};
				}
				
				\foreach \x in {1,...,11} {
					\draw (4*\x,5) node[voter] {};
				}
			\end{tikzpicture}
			\caption{If the positions of all proxies are on the same side of $\tmed$, then non of them have a manipulation.\vspace{-4mm}}
		\end{figure}
	\end{proof}

	
	\section{Manipulations for Better Outcomes}
	\label{sec:proxy_manipulation_better_outcomes}
	
	Consider the manipulation described in the proof of Theorem~\ref{thm:manipulable_proxy_characterization}. The outcome of the manipulated state is the true median, which is the outcome of the median voting rule with complete participation. That is, in this case the manipulation has a positive effect on the accuracy of the outcome!
	
	\paragraph{Manipulations are often beneficial.}
	The example above may seem counter-intuitive, but it is in fact common that strategic behavior improves the outcome, even when applied repeatedly by all voters. This is true especially in simple voting rules like Plurality, as manipulations play a form of compromise that lets voters avoid socially-inferior outcomes. This was shown both in theoretical analysis~\cite{grandi2013restricted} and in simulations~\cite{meir2014local,grandi2013restricted}.
	
	\paragraph{Evaluation.}
	This brings about the natural question \emph{does strategic voting of proxies always improve the outcome?}
	We note that in general, questions about `good outcomes' in voting are tricky, since there are numerous ways to evaluate the outcome of a voting rule, e.g. if we define the optimum to be the outcome of the used rule itself on the truthful votes (as in \cite{branzei2013bad}), then strategic behavior is bad by definition. 
	
	However, in the context of delegation on the real line, a natural evaluation metric is the distance to the `ideal point' that would be selected if everyone had voted, i.e. to $\statemed{p}$. This is exactly the approach followed by \citep{cohensius2016proxy}, when showing (non-strategic) delegation improves the outcome. 
	
	Another natural measure is the \emph{social cost}, i.e. the sum of distances of all voters from the selected position. When using the median voting rule without delegation, the two goals coincide, as the median is known to minimize the social cost. However a good approximation of the true median may have poor social cost, and vice versa.
	\footnote{Suppose there are $k$ followers and 1 proxy $j$ on 0, $k+2$ followers on 1, and a second proxy $j'$ on $2-\varepsilon$. Then $j'$ is the closest proxy to $\tmed=1$ (and thus selected by WM), but has social cost of $3k+\Theta(1)$, whereas $j$ has social cost of $k+\Theta(1)$.
		
		Conversely, if $j'$ is on $1+2/k$, then $j$ still minimizes the social cost, but its distance from $\tmed$ is $\Theta(k)$ larger than that of $j'$. 
	}
	
	Following \cite{cohensius2016proxy}, we adopt the distance to the true median as our goal, but also discuss the implications on social cost where relevant.
	
	\paragraph{Convergence.}
	When discussing strategic behavior, an even more fundamental question than welfare is \emph{stability}. A hierarchy for notions of stability in iterative voting is explained in \cite{meir2018strategic}. For example, under Plurality with a mild assumption on voters' behavior it is known that iterative voting always converges to a pure Nash equilibrium~\cite{meir2017iterative}. In candidacy games, which are equivalent to our setting with weightless proxies, it was shown that a pure Nash equilibrium exists~\cite{sabato2017real}, but there are no results regarding convergence, or regarding the more general model where proxies have weights.
	
	In this section and the next one, we therefore study both the conditions under which iterative voting by strategic proxies is guaranteed to converge, and bounds on the distance of the final outcome from the true median of the population.

	\subsection{Dynamics and Convergence}
	
	\paragraph{Policies.}
	A \emph{policy} for proxy $j \in M$ is a function that maps a state to a strategy. Formally, let $\cS=\R^m$ be the set of all possible states for the proxies, then, a policy for $j$ is a function $\pi_{j}: \cS \to \R$.
	
	\paragraph{Better-responses.}
	For every $j \in M$ and every state $s$, we say that the position $s_{j}'$ is a \emph{better-response} to $s$ if $j$ strictly prefers the outcome of $\parenthesis{s_{-j},s_{j}'}$ to the outcome of $s$. We denote the set of better-responses of $j$ to $s$ by $\better{s}{j}$. We say that a policy is a \emph{better-response policy} for $j$ if for every $s$, the strategy selected by the policy is a better response to $s$, that is $\pi_j\parenthesis{s} \in \better{s}{j}$. A better-response policy is said to be \emph{best-response policy}, if the outcome of the selected strategy is most preferred by the proxy within their better-response set. Note that a best-response policy may not exist.
	
	\paragraph{Truth-oriented.}
	A proxy $j$ is \emph{truth-oriented} if their policy selects their true position whenever it is in their better-response set, and is weakly better than any other strategy. Formally, for every $s$, if $p_{j} \in \better{s}{j}$ and $j$ prefers $\wm{s_{-j},p_j}$ to $\wm{s_{-j},s_j'}$ for every $s_j' \in \better{s}{j}$, then $\pi_{j}\parenthesis{s}=p_{j}$. 
	Truth-orientation is closely related to \emph{truth-bias} proposed by ~\citet{meir2010convergence}. Truth-biased agents would resort to truth if it is weakly better than \emph{any} other strategy, in particular when the better-response set is empty. Truth-orientation is a weaker requirement, as truth is only compared to better-responses.
	
	\paragraph{Dynamics.}
	A \emph{dynamics} $\dyn{s} = \parenthesis{s^t}_{t=0}^{\infty}$ is a (possibly infinite) series of states, where $s^t$ is the state after step $t$. We assume that the initial state is truthful, i.e., $s^0 = p|_M$.
	
	Then, for every $t > 0$  there is a single proxy $j=\pmov{t} \in M$ changing position from $s^{t-1}_j$ to $s^t_{j}$ according to their policy $\pi_j$. Thus
	$$s^t = \parenthesis{s^{t-1}_{-j},s^t_j} = \parenthesis{s^{t-1}_{-j}, \pi_j\parenthesis{s^{t-1}}}.$$
	We do not assume any particular order over proxies' turns, except that there is no starvation. That is,  every proxy eventually gets to play again, an infinite number of times.
	
	Recall that the winner in state $s$ is denoted by $\pwin{s}$. We denote by $\pwin{t}$ and $\swm{t}$ the winner at time $t$ and their position, respectively.  Thus, $\pwin{t} = \pwin{s^t}$ and $\swm{t}=s^t_{j^*(t)}=\wm{s^t}$. We further denote by $\med{t} := \statemed{s^t}$ the median at step $t$.
	
	We further denote by  by $j^*:=j^*(s^0)$ the winner of the initial state, thus $p_{j^*}=\wm{p}=\twm$.
	
	At a given state $s$, we denote by $\Delta(s):=|\statemed{s}-\wm{s}|$ the distance between the unweighted median and the weighted median (the winner). We also denote $\Delta^t:=\Delta(s^t)$ and $\Delta:=\Delta(s^0)$. 
	
	The standard setting for the study of on-going dynamics in voting is Iterative Voting~\citep{meir2017iterative}. However, since our model involves an infinite action set, the terminology and results cannot be applied in a straightforward way. We address it when relevant. Instead, we say that a dynamics $\dyn{s}$ \emph{converges} if it has a limit.
	
	A state $s$ is a \emph{pure Nash equilibrium} (PNE) if for every $j \in M$ it holds that $\better{s}{j} = \varnothing$, that is, no proxy has a better-response to $s$.
	
	We start our analysis with bounding the distance from the true median that the outcome can converge to.  For the rest of this section, we show that for every step in a better-response dynamics from truth, the current median and outcome are bounded in a neighborhood of $\tmed$ with radius $\Delta$. 
	
	\begin{theorem}
		\label{thm:bounded_proxies}
		Assume all proxies are truth-oriented. For every step $t\ge 0$ in a better-response dynamics $\dyn{s}$ by proxy~$j$, if $j$'s peak is left of $\tmed$ then $s^t_{j} \le \tmed+\Delta$; and if their peak is right of $\tmed$, then $\tmed-\Delta \le s^t_{j}$.
	\end{theorem}
	\begin{remark}
		We point out that if proxies are weightless, as in the model by \citet{sabato2017real}, then the theorem trivially holds. This is since $\statemed{s}=\statemed{p}$ at every state, and thus $\Delta^t$ can only become smaller at every step, as the current proxy moves closer to $\statemed{p}$ to become the new winner.
	\end{remark}

	Consider a truthful state where the peak of the initial winner $p_{j^*}$ is left of the median. For a better-response dynamics from this truthful state, we prove the following lemma:
	
	\begin{lemma}
		\label{lm:truthful_winner_cant_drift}
		Assume all proxies are truth-oriented. If for every $0 \le t' \le t$, and for all $j\in M$ s.t. $p_j\geq \statemed{p}$ it holds that $s^{t'}_j\geq med-\Delta=p_{j^*}$, then for the truthful winner $j^*$ it holds that $p_{j^*} \le s^t_{j^*}$
	\end{lemma}
	\begin{proof}
		Assume false, and let $0 < t_0 \le t$ be the first step that the truthful winner $j^*$ moves to the left of their peak, that is, $s^{t_0}_{j^*} < p_{j^*} \le s^{t_0-1}_{j^*}$. In particular this means that $j^*$ is the moving proxy at $t_0-1$.
		Since no proxy with peak right of $\tmed$ reported a position left of $\tmed-\Delta$ prior to $t$, it follows that the median at every step prior to $t$ is right of $\tmed-\Delta$. This is in particular true at $t_0\le t$, that is, $\tmed-\Delta \le \med{t_0}$. It follows that $s^{t_0}_{j^*} < p_{j^*} = \tmed-\Delta \le \med{t_0}$.
		
		We argue that this contradicts the assumption that proxies are truth-oriented. One of the following must hold. Either $s^{t_0}_{j^*}$ is the winning position at $t_0$, in which case $p_{j^*}$ is closer to the current median and can therefore win. Since the peak is the optimal outcome for $j^*$, it is in their better-response set, the outcome is weakly better than the outcome of every other strategy in $\better{s^{t_0}}{j^*}$. Otherwise, since $s^{t_0}_{j^*}$ is a better-response to $s^{t_0-1}$, it must be that by reporting $s^{t_0}_{j^*}$ the position of the median changed, such that $\pwin{t_0}\ne j^*$ is closer to $\med{t_0}$ than $\med{t_0-1}$. Again since $t_0,t_0-1\le t$ we get that $p_{j^*} \le \med{t_0},\med{t_0-1}$. Thus, reporting $p_{j^*}$ would have the same effect on the position of the median as $s^{t_0}_{j^*}$. Therefore, $p_{j^*}$ is weakly better than any other better-response. Figure~\ref{fig:truthful_winner_cant_drift} demonstrates the possible scenarios.\qed
		
		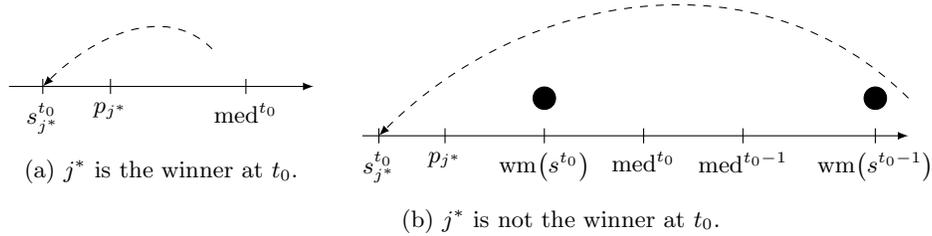
\begin{figure}[ht]
			\vspace{-2em}
			\begin{subfigure}{0.35\textwidth}
				\centering
				\begin{tikzpicture}[x=0.9cm, y=1cm, z=1mm]
					\coordinate (s) at (1,0);
					\coordinate (p) at (2,0);
					\coordinate (med) at (4,0);
					
					\draw[-latex] (0.5,0) -- (5,0);
					\draw (s)++(0,0.1) -- ++(0,-0.2) node[below] {$s^{t_0}_{j^*}$};
					\draw (p)++(0,0.1) -- ++(0,-0.2) node[below] {$p_{j^*}$};
					\draw (med)++(0,0.1) -- ++(0,-0.2) node[below] {$\med{t_0}$};
					
					\draw[dashed,-latex] (med)++(-0.5,0.5) to[out=135,in=45] (s);
				\end{tikzpicture}
				\caption{$j^*$ is the winner at $t_0$.}
			\end{subfigure}
			\quad
			\begin{subfigure}{0.45\textwidth}
				\begin{tikzpicture}[x=0.88cm, y=1cm, z=1mm]
					\coordinate (s) at (1,0);
					\coordinate (p) at (2,0);
					\coordinate (wa) at (3.5,0);
					\coordinate (meda) at (5,0);
					\coordinate (medb) at (6.5,0);
					\coordinate (wb) at (8.5,0);
					
					\draw[-latex] (0.75,0) -- (9,0);
					\draw (s)++(0,0.1) -- ++(0,-0.2) node[below] {$s^{t_0}_{j^*}$};
					\draw (p)++(0,0.1) -- ++(0,-0.2) node[below] {$p_{j^*}$};
					\draw (wa)++(0,0.1) -- ++(0,-0.2) node[below] {$\swin{t_0}$};
					\draw (meda)++(0,0.1) -- ++(0,-0.2) node[below] {$\med{t_0}$};
					\draw (medb)++(0,0.1) -- ++(0,-0.2) node[below] {$\med{t_0-1}$};
					\draw (wb)++(0,0.1) -- ++(0,-0.2) node[below] {$\swin{t_0-1}$};
					
					\draw[dashed,-latex] (wb)++(0.5,0.5) to[out=135,in=45] (s);
					
					\draw (wa)+(0,0.5) node[circle,draw=black,fill,inner sep=3pt] {};
					\draw (wb)+(0,0.5) node[circle,draw=black,fill,inner sep=3pt] {};
				\end{tikzpicture}
				\caption{$j^*$ is not the winner at $t_0$.}
			\end{subfigure}
			\caption{Possible states at $t_0$ after $t^*$ moves.}
			\label{fig:truthful_winner_cant_drift}
		\end{figure}
		\vspace{-1em}
	\end{proof}
	
	Note that by symmetry the same hold for the case where the truthful winner's peak is right of the median. We turn to prove Theorem~\ref{thm:bounded_proxies}.
	
	\begin{proof}
		Assume that there was no violation up to step $t\ge 0$. In particular, this imply that $\tmed-\Delta \le \med{t} \le \tmed+\Delta$.
		
		Consider a proxy $j$, and let $s_{j}'$ be a possible strategy for $j$ such that \\$s_{j}' \notin \brackets{\tmed-\Delta, \tmed+\Delta}$. If $s_{j}'$ is on the same side of the median as $p_j$, then this is not a violation. Otherwise, $s_{j}'$ is a better-response iff it is between $p_j$ and $\swin{t}$, the winning position at $t$. This is because otherwise, $s_{j}'$ would either have no effect on the outcome, or the outcome would be further than $\swin{t}$ to their peak. Note that this also implies that $\swin{t} \notin \brackets{\tmed-\Delta,\tmed+\Delta}$. By lemma~\ref{lm:truthful_winner_cant_drift}, we have that $s^t_{j^*} \in \brackets{\tmed-\Delta,\tmed+\Delta}$. Thus, $\swin{t} \notin \brackets{\tmed-\Delta,\tmed+\Delta}$ only if $\swin{t}$ and $s^t_{j^*}$ are not on the same side of the median. Moreover, it must be that $\med{t}$ is between $\tmed$ and the position of the current winner $\swin{t}$. However, if $\med{t} \ne \tmed$, then there is a proxy with a reported position in $s^t$ that is between $\med{t}$ and $\swin{t}$ such that their peak is on the other side of $\tmed$ than $\med{t}$. This contradicts Lemma~\ref{lm:nearest_proxy_to_median_wins}, thus, all violations are not in the better-response set of the proxies, and there is no violation at step $t$.\qed
		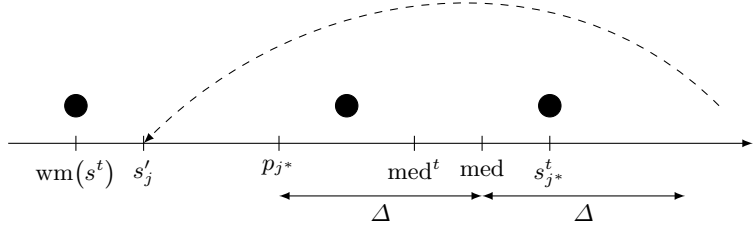
\begin{figure}[ht!]
			\centering
			\begin{tikzpicture}[x=0.9cm, y=1cm, z=1mm]
				\coordinate (s) at (1,0);
				\coordinate (st) at (2,0);
				\coordinate (p) at (4,0);
				\coordinate (d) at (5,0);
				\coordinate (medt) at (6,0);
				\coordinate (med) at (7,0);
				\coordinate (ow) at (8,0);
				\coordinate (dt) at (10,0);
				
				\draw[-latex] (0,0) -- (dt)--++(1,0);
				\draw (s)++(0,0.1) -- ++(0,-0.2) node[below] {$\swin{t}$};
				\draw (st)++(0,0.1) -- ++(0,-0.2) node[below] {$s_{j}'$};
				\draw (p)++(0,0.1) -- ++(0,-0.2) node[below] {$p_{j^*}$};
				\draw (medt)++(0,0.1) -- ++(0,-0.2) node[below] {$\med{t}$};
				\draw (med)++(0,0.1) -- ++(0,-0.2) node[below] {$\tmed$};
				\draw (ow)++(0,0.1) -- ++(0,-0.2) node[below] {$s^t_{j^*}$};
				
				\draw[dashed,-latex] (dt)++(0.5,0.5) to[out=135,in=45] (st);
				
				\draw (s)+(0,0.5) node[circle,draw=black,fill,inner sep=3pt] {};
				\draw (d)+(0,0.5) node[circle,draw=black,fill,inner sep=3pt] {};
				\draw (ow)+(0,0.5) node[circle,draw=black,fill,inner sep=3pt] {};
				
				\draw[latex-latex] (p)++(0,-0.7) -- node[below] {$\Delta$} ++(3,0);
				\draw[latex-latex] (med)++(0,-0.7) -- node[below] {$\Delta$} ++(3,0);
			\end{tikzpicture}
			\caption{Proof of Theorem~\ref{thm:bounded_proxies}. For a violation to be a better-response, this must be the state at $t$, in contradiction to Lemma~\ref{lm:nearest_proxy_to_median_wins}.}
		\end{figure}
	\end{proof}
	
	\begin{corollary}
		For every state $s^t$ in a better-response dynamics $\dyn{s}$ with truth-oriented proxies, both the median and the outcome of $s^t$ are in the interval $\brackets{\tmed-\Delta,\tmed+\Delta}$.
	\end{corollary}
	
	The bound on the outcome shows that strategic behavior in proxy voting can reduce the distance between the outcome and the true median. However, it does not convergence to a stable state (equilibrium), or even a reduced social cost. In what follows, we discuss conditions for both.
	
	\subsection{Monotone Policies}
	
	\paragraph{Monotonicity justification}
	Consider the better-response set of some proxy $j$ with peak $p_j < \statemed{s}$ at state $s$. While it is possible that $s_j$ is on the same side of $\statemed{s}$ and there are better responses on both sides, the following must hold:
	\begin{itemize}
		\item There is at least one better response $s'_j\leq \statemed{s}$;
		\item $j$ weakly prefers any better-response $s'_j\leq \statemed{s}$ to any $s''_j>\statemed{s}$, due to single-peakedness.
	\end{itemize}
	Therefore, it is reasonable to assume that proxies 
	restrict their policies so as to select a position that is on the same side of the median as their true position.\footnote{This assumption is somewhat similar to a `no overbidding' assumption in auctions.} In the following discussion, we restrict policies to ones that preserves the integrity of proxies positions with respect to the median.
	
	\paragraph{Monotonicity}
	Formally, we say that a better-response dynamics $\dyn{s}$ is \emph{monotone} if for every $j \in M$, we have for any $j\in M$ s.t. $p_{j} \le \tmed$, and any step $t$, that $\pi_{j}\parenthesis{s^t} \le \tmed$ (and likewise for $p_j \ge \tmed$. Note that for every state $s^t$ of a monotone better-response $\dyn{s}$, the median of $s^t$ is $\tmed$.
	
	\begin{observation}
		Under monotone dynamics, $\med{t}=\tmed$ for all $t$.
	\end{observation}
	
	This is since the same set of voters (followers and proxies) remain on each side of the $\tmed$.
	
	\subsection{Narrowing in on the Median}
	Our goal in this section is to prove that that any monotone better-response dynamics converges to the true median. The problem is that this may not hold at every step, which requires some extra work.

	The following Lemma shows that any better-response in a monotone better-response dynamics where the winning proxy is not the moving proxy strictly decreases the distance to the median.
	\begin{lemma}
		\label{lm:not_winner_manipulation_closer_to_median}
		Let $\dyn{s}$ be a monotone better-response dynamics. Then, for every $t \ge 0$ if $\pmov{t+1} \ne \pwin{t}$, then $\Delta^{t+1} < \Delta^t$.
	\end{lemma}
	\begin{proof}
		By Lemma~\ref{lm:nearest_proxy_to_median_wins}, for every $k \in M$ it holds that
		$$\abs{\swin{t+1}-\tmed} \le \abs{s^{t+1}_{k}-\tmed}.$$ In particular, this holds for $\pwin{t} \in M$. We get:
		\begin{equation*}
			\abs{\swin{t+1}-\tmed} \le \abs{s^{t+1}_{\pwin{t}}-\tmed}
		\end{equation*}
		
		Since $\pmov{t+1} \ne \pwin{t}$, it holds that $\smov{t+1}$ is a better-response for $s^t$, so\\ $\abs{s^{t+1}_{\pwin{t+1}}-\tmed} \ne \abs{s^{t+1}_{\pwin{t}}-\tmed}$. Hence:
		\begin{align*}
			\Delta^{t+1} = \abs{\swin{t+1}-\tmed} < \abs{\swin{t}-\tmed} = \Delta^t.~~~~~~~~~~
			\qed
		\end{align*}
	\end{proof}
	
	While Lemma~\ref{lm:not_winner_manipulation_closer_to_median} shows that moves made by proxies with reported position not at the current outcome must reduce the distance to the true median, it is possible for winning proxies to move in a way that increase the distance to the median. Figure~\ref{fig:winners_increase} describe a proxy that makes 2 consecutive steps. The first makes them the winning proxy, the next is a better-response. As they remain the winning proxy, the outcome after the second step is further from the median.
	
	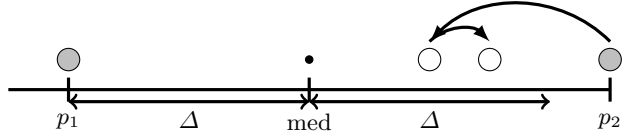
\begin{figure}[t]
		\centering
		\begin{tikzpicture}[
			line/.style={draw=black, very thick},
			proxy/.style={circle,draw=black,fill,inner sep=\proxyradius pt},
			emptyproxy/.style={circle,draw=black,fill=none,inner sep=\proxyradius pt},
			voter/.style={circle,draw=black,fill,inner sep=\voterradius pt},
			x=0.8mm, y=0.8mm, z=1mm]
			
			\draw[line, -](-50,0)--(50,0);
			
			\foreach \x/\proxyind/\t in {-4/1/p_{1},0/0/\tmed,5/1/p_{2}} {
				\draw[line](10*\x,-2)--(10*\x,2) node[below=10]{$\t$};
				\ifthenelse{\proxyind=1}
				{\draw (10*\x,5) node[circle,draw=black,fill=gray!50,inner sep=\proxyradius pt] {};}
				{\draw (10*\x,5) node[voter] {};}
			}
			
			\foreach \x in {2,3} {
				\draw (10*\x,5) node[emptyproxy] {};
			}
			
			\draw[line, <->](10*-4,-2)--(0,-2) node at (10*-2,-2)[below]{$\Delta$};
			\draw[line, <->](0,-2)--(10*4,-2) node at (10*2,-2)[below]{$\Delta$};
			
			\draw[line,-latex] (5*10,8) to[out=135,in=45] (2*10,8);
			\draw[line,-latex] (2*10,8) to[out=45,in=135] (3*10,8);

		\end{tikzpicture}
		\caption{Consecutive steps that increase the distance to the median. Gray dots indicate truthful positions of proxies, empty dots indicate positions of manipulation. Arrows indicate moves. A small full dot is the position of a (single) follower.}
		\label{fig:winners_increase}
		\vspace{-1em}
	\end{figure}
	
	\paragraph{Meta-moves.}
	We call sequences of consecutive better-responses by the same winning proxy \emph{meta-move}.
	Formally, a meta-move of length $\ell$ from $s^t$ is a subsequence of steps in a better-response dynamics $\dyn{s}$ such that:
	\begin{itemize}
		\item $\pmov{t+1} \ne\pwin{t}$ and $\pwin{t+1}=\pmov{t+1}$. That is, in state $s^t$, the proxy $\pmov{t}$ moves in a way that makes them the winner.
		\item Let $\ell > 0$ such that for every $1 \le i \le \ell$ it holds that $\pmov{t+i} = \pmov{t+1} = \pwin{t+1}$. In other words, after $\pmov{t+1}$ becomes the winning proxy at step $t+1$, they continue to make consecutive better-responses for $\ell$ steps.
	\end{itemize}
	
	The following shows that while local manipulations within a meta-move can increase the current distance to the true median (as Figure~\ref{fig:winners_increase} demonstrates), meta-moves globally decrease the distance to the true median.
	\begin{lemma}
		\label{lm:winner_manipulation_closer_to_median}
		For every meta-move of length $\ell$ from $s^t$ of a monotone better-response dynamics $\dyn{s}$, it holds that $\Delta^{t+\ell} < \Delta^t$.
	\end{lemma}
	\begin{proof}
		By Lemma~\ref{lm:nearest_proxy_to_median_wins}, monotonicity and since for every $1 \le i \le \ell$ it holds that $\pmov{t+i}=\pwin{t+1}\ne\pwin{t}$, we get that $\Delta^{t+i} \le \Delta^t$. Furthermore, since $\smov{t+1}$ is a better-response for $\pmov{t+1}$, it must be that the outcome of $s^{t+1}$ is not equal to the outcome of $s^t$. We get that for every $i$, $\smov{t+i}$ is a better-response and therefore $\smov{t+i}\ne \smov{t+1}$. Thus $\Delta^{t+i}\ne\Delta^t$. In particular, this holds for $i=\ell$.\qed
	\end{proof}
	
	Lemma~\ref{lm:not_winner_manipulation_closer_to_median} and Lemma~\ref{lm:winner_manipulation_closer_to_median} together provide a complete analysis of the better-response sets for proxies, and show that the better-response set strictly decreases after each (meta) move. However, this alone is not sufficient for convergence.
	
	\begin{example}
		\label{xmpl:dynamics_diverge}
		Recall the setting appearing in Example~\ref{xmpl:perm}. Define $\alpha_1 = \frac{1}{4}$, and for every $t \in \N$, define $\alpha_{t+1}=\frac{1}{2}\alpha_t$. We define the following policy for $j \in M$:
		\begin{equation*}
			\pi_{j}\parenthesis{s^t} = \tmed - sign\parenthesis{\tmed - p_{j}}\parenthesis{\Delta^t - \alpha_t}
		\end{equation*}
		
		For every $t \in \N$ we get that
		\begin{align*}
			\Delta^{t+1}&=\abs{\swin{t+1} - \tmed}\\
			&= \abs{\tmed - sign\parenthesis{\tmed - p_{\pwin{t+1}}}\parenthesis{\Delta^t - \alpha_t} - \tmed}\\
			&= \abs{-sign\parenthesis{\tmed - p_{\pwin{t+1}}}\parenthesis{\Delta^t - \alpha_t}} = \Delta^t - \alpha_t
		\end{align*}
		As $\alpha_t = \frac{1}{2}\alpha_{t-1}$, we get $\Delta^{t+1} = \Delta^1 - \sum_{i=0}^{t-2} \frac{1}{2^i}\alpha_1 = \Delta^1 - \alpha_1\sum_{i=0}^{t-2} \frac{1}{2^i}$. As $t \to \infty$, we get that the distance to the median converges to $\Delta^1-2\alpha_1 = \Delta^1-2\frac{1}{4}\Delta^1 = \frac{1}{2}\Delta^1$, and the outcome oscillates between $-\frac{1}{2}$ and $\frac{1}{2}$. Thus the best-response dynamics diverges. Figure~\ref{fig:diverge_example} shows a schematic of this dynamic.
		
		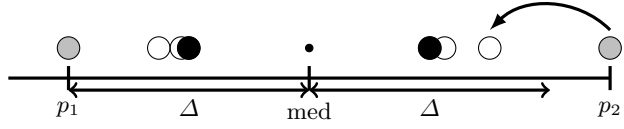
\begin{figure}[t]
			\centering
			\begin{tikzpicture}[
				line/.style={draw=black, very thick},
				proxy/.style={circle,draw=black,fill,inner sep=\proxyradius pt},
				emptyproxy/.style={circle,draw=black,fill=none,inner sep=\proxyradius pt},
				voter/.style={circle,draw=black,fill,inner sep=\voterradius pt},
				x=0.8mm, y=0.8mm, z=1mm]
				
				\draw[line, -](-50,0)--(50,0);
				
				\foreach \x/\proxyind/\t in {-4/1/p_{1},0/0/\tmed,5/1/p_{2}} {
					\draw[line](10*\x,-2)--(10*\x,2) node[below=10]{$\t$};
					\ifthenelse{\proxyind=1}
					{\draw (10*\x,5) node[circle,draw=black,fill=gray!50,inner sep=\proxyradius pt] {};}
					{\draw (10*\x,5) node[voter] {};}
				}
				
				\foreach \x in {-2.5,-2.125,-2-1/32,2+1/16,2.25,3} {
					\draw (10*\x,5) node[emptyproxy] {};
				}
				
				\draw (10*2,5) node[proxy] {};
				\draw (10*-2,5) node[proxy] {};
				
				\draw[line, <->](10*-4,-2)--(0,-2) node at (10*-2,-2)[below]{$\Delta$};
				\draw[line, <->](0,-2)--(10*4,-2) node at (10*2,-2)[below]{$\Delta$};
				\draw[line,-latex] (5*10,8) to[out=135,in=45] (3*10,8);

			\end{tikzpicture}
			\caption{A dynamics that diverges.
				The two large black  dots indicate oscillation positions.
				The arrow indicates the first manipulation.
				\vspace{-3mm}}
			\label{fig:diverge_example}
		\end{figure}
	\end{example}
	
	\paragraph{Iterative Voting comparison.}
	Note that Example~\ref{xmpl:dynamics_diverge} not only shows that monotone better-response dynamics need not converge, it also shows a key difference between our setting and Iterative Voting. We say that a dynamic is \emph{acyclic} if there are no recurring states. For finite action sets, i.e., when the space of available better-responses for each agent is finite, acyclicity implies convergence. Example~\ref{xmpl:dynamics_diverge} demonstrates that for infinite action spaces this may not hold. 
	
	\paragraph{Interpretation of alpha.}
	In effect, $\alpha_t$ is the amount by which the outcome gets closer to the true median between steps. As $\Delta^t$ decreases, so does the leeway that proxies have to improve the outcome for themselves. While it is reasonable that $\alpha_t$ decreases as $\Delta^t$ decreases, Example~\ref{xmpl:dynamics_diverge} captures the behavior in which $\alpha_t$ decreases at a higher rate than $\Delta^t$.
	
	By restricting policies such that $\alpha_t$ and $\Delta^t$ decrease at the same rate, we can obtain convergence. Moreover, this guarantees that $\Delta^t$ itself converges to $0$, meaning that the outcome converges to the true median.
	
	While the example above shows that even monotone policies may diverge, this relies one the gaps between $\Delta^t$ becoming smaller and smaller. Fix some constant $\alpha<1$. We say that a meta step is \emph{big} if $\Delta^{t+\ell}<\alpha \Delta^{t}$, and otherwise it is \emph{small}.
	\begin{corollary}
		\label{corr:sufficient_step_implies_convergence}
		Consider a monotone better-response dynamics, and
		suppose there is only a finite number of small steps between any two big steps. Then the dynamics converges, and the limit PNE is the true median.
	\end{corollary}
	\begin{proof}
		After $T$ big meta steps, we have that $|\tmed-\wm{s}|<\alpha^T \Delta \rightarrow 0$.\qed
	\end{proof}
	Moreover, the corollary still holds if $\alpha$ is not a constant but increases towards $1$ over time, as long as this does not occur too fast.\footnote{For example, we can allow $\alpha_k = \max(\alpha,1-\frac{1}{k})$, at the $k$'th big meta step.}
	\paragraph{Why are there big steps?}
	We argue that it is not reasonable that proxies will insist on small steps forever as in Example~\ref{xmpl:dynamics_diverge}. 
	
	To see why, note that while smaller steps are preferable to the moving proxy, this benefit get smaller and smaller. On the other hand,  the fraction of small meta-steps from all better-responses becomes smaller over time, so almost every better-response is big. 
	
	\paragraph{Why this result is good.}
	The true median is the outcome of the median voting rule. It is both Condorcet consistent and the minimal sum of distances from voters' true preferences. Thus, the median of all voters reflects the social optimum. As such, Corollary~\ref{corr:sufficient_step_implies_convergence} implies that the strategic behavior of proxies (with the above restrictions) can in fact produce a socially optimal and stable outcome.
	
	\subsection{Discretization}
	
	\paragraph{Justify discretization}
	In many real-world applications, the assumption that voters can express any position on the political spectrum $\R$ is unreasonable. Voters are unlikely to distinguish between positions that are too similar, and this is the case both for selecting their truthful position, and distinguishing between different proxy positions for delegation. In computerized settings, there is some limited resolution to the expression of preferences (e.g. a temperature or a monetary amount). As it turns out, any such limit eliminates the possibility of oscillation we encountered in the previous section.
	
	In this section, we assume w.l.o.g that the political spectrum is restricted to the set of all integers $\Z$.
	
	\paragraph{Convergence for discrete spaces.}
	For discrete spaces, every monotone policy meets the conditions of Corollary~\ref{corr:sufficient_step_implies_convergence}. This is due to the fact that every better-response made by a proxy with position that is not the current weighted median must decrease the distance to the true median by at least $1$ (as the minimal distance between every distinct possible positions). Thus, the conditions are met for $\alpha = 1-\frac{1}{\Delta^1}$. Therefore, for discrete spaces, every monotone better-response dynamics converges, and the outcome is the true median, which is the socially optimal outcome.
	
	\paragraph{Best response.}
	Furthermore, for discrete spaces (in contrast to continuous) there is a well-defined best-response, that is to reposition at a distance that is one step closer to the true median than the current winner on their opposite side of the median. In particular, the best-response is monotone.
	
	\paragraph{Connection to Iterative Voting.}
	Following the terminology of~\citep{meir2017iterative}, a game has the \emph{Finite Best Response Property (FBRP) from truth} if from any truthful state, when restricted to best-responses, the dynamics converges. Thus, SPGs with WM are FBRP from truth.

	However, for non-monotone policies convergence to a socially worse outcome is possible. We show tis in Appendix~\ref{appx:worse_sc_discrete}.
	
	Our conjecture is that convergence holds for the non-discrete case as well, and that ultimately proxies would have an incentive to deviate back to their original side of the median. Yet, this is a matter of future research.
	
	
	\section{Partial Information}
	\label{sec:partial_information}
	
	In previous sections we assumed that the proxies have complete information about the positions of proxies and followers alike. This assumption is common when analyzing adversarial behavior. However, is it reasonable in a proxy voting setting?
	
	Recall that one of the applications of proxy voting is to mitigate the adverse effects of partial participation, where voters want to avoid explicitly reporting their positions. Moreover, followers may not even know their exact position; instead, they only know how to rank proxies based on proximity. Thus, followers can still delegate their vote without the additional cognitive strain of determining their exact position.
	
	In this section we relax the assumption of complete information. First, it is worth noting that when proxies have no information about the positions of followers, proxy voting becomes strategyproof. To see this, consider the states appearing in Figure~\ref{fig:partial_indistinguishable}
	
	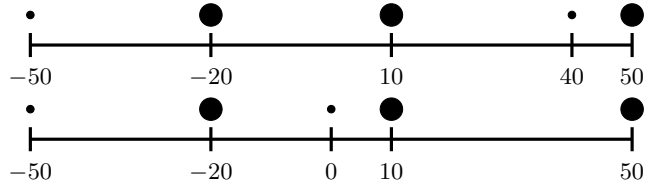
\begin{figure}[h]
		\centering
		\begin{tikzpicture}[
			line/.style={draw=black, very thick},
			proxy/.style={circle,draw=black,fill,inner sep=\proxyradius pt},
			emptyproxy/.style={circle,draw=black,fill=none,inner sep=\proxyradius pt},
			voter/.style={circle,draw=black,fill,inner sep=\voterradius pt},
			x=0.8mm, y=0.8mm, z=1mm]
			
			\draw[line, -](-50,0)--(50,0);
			
			\foreach \x/\t in {-50/0,40/0,-20/1,10/1,50/1} {
				\draw[line](\x,-2)--(\x,2) node[below=10]{$\x$};
				\ifthenelse{\t=1}
				{\draw (\x,5) node[proxy] {};}
				{\draw (\x,5) node[voter] {};}
				
			}
		\end{tikzpicture}
		\begin{tikzpicture}[
			line/.style={draw=black, very thick},
			proxy/.style={circle,draw=black,fill,inner sep=\proxyradius pt},
			emptyproxy/.style={circle,draw=black,fill=none,inner sep=\proxyradius pt},
			voter/.style={circle,draw=black,fill,inner sep=\voterradius pt},
			x=0.8mm, y=0.8mm, z=1mm]
			
			\draw[line, -](-50,0)--(50,0);
			
			\foreach \x/\t in {-50/0,-20/1,0/0,10/1,50/1} {
				\draw[line](\x,-2)--(\x,2) node[below=10]{$\x$};
				\ifthenelse{\t=1}
				{\draw (\x,5) node[proxy] {};}
				{\draw (\x,5) node[voter] {};}
				
			}
		\end{tikzpicture}
		\caption{An example of two states that are indistinguishable if followers' positions are unknown to proxies.\vspace{-3mm}}
		\label{fig:partial_indistinguishable}
	\end{figure}
	
	In the bottom state, the proxy at $-20$ can manipulate the outcome by deviating to $-5$. However, in the top state, the proxy has no manipulation. When proxies have no information except proxies positions, the proxies cannot distinguish between the two states. Thus, proxies do not even know \emph{if} they have a valid manipulation, let alone find one.
	
	However assuming no information at all sounds too restrictive, and we would like to consider intermediate cases that are more reasonable.
	
	\paragraph{Outline summary for section.}
	For the rest of this section, we first describe formally a less restrictive setting for the study of partial information. Then, we show that when only partial information is made available to voters, the strategic behavior of proxies may converge to a worse social outcome than the truthful state.
	
	\subsection{Model}
	\paragraph{Information sets.}
	We employ the framework described in~\citep{reijngoud2012voter}. In this setting, a \emph{poll information function (PIF)} $\sigma$ maps each state $s$ to an information set $\sigma\parenthesis{s}$. For example, in a Plurality voting scenario, we can think of a PIF that returns the score of each candidate, or just the candidate ranking, or even just the name of the winning candidate.
	The set $\sigma\parenthesis{s}$ is then communicated to all voters.

	Intuitively, we can think of the PIF as the results of a poll that is broadcast publicly after all private information is collected. 
	
	\paragraph{Poll information with delegation.}
	What information is likely to become public in our setting? Clearly the proxies' positions, as otherwise followers will not be able to delegate. 
	Other than that, we only assume that the identity (and position) of the winner is announced. 
	
	In this section only we use $p$ instead of $p|_N$ for the followers' positions, to simplify notation. To avoid confusion, we use $s^0$ rather than $p|_M$  for proxies' true positions. 
	
	We thus denote by $\sigma_{winner}$ the PIF that takes as input the state $s$ (proxies' current positions) and followers' true positions $p$ and returns $(s,\wm{s,p})$. That is, reveals the proxies' positions and the winner.
	Since this is the only PIF we use in this work, we just write $\sigma$. 
	
	In this setting, proxies are unable to distinguish between states that yield the same information by $\sigma$. Recall the two states from Figure~\ref{fig:partial_indistinguishable}. When only proxy positions are communicated by $\sigma$, the states are indistinguishable by the proxies.
	However, proxies can deduce an equivalent set of states that are consistent with the information available to them. In particular, both states in Figure~\ref{fig:partial_indistinguishable} would be in the same set.

	Formally, at any state $s$, the proxies know only the identity of the winner $j^*$, and their positions. We define the set of possible profiles as 
	$$\IS{s}{j^*}:= \{p\in \mathbb{R}^n: \sigma(s,p)=(s,j^*)\}.$$
	These are all possible positions of followers that are compatible with the revealed information.
	
	While in the previous sections $j^*=j^*(s)$ could be implicitly inferred from proxies' positions $s$ (since $p$ was known), in this section the state is defined as $(s,j^*)$, i.e. it explicitly contains all known information. 
	
	\paragraph{Dominating manipulations.}
	Following the terminology of~\citep{conitzer2011dominating}, we define a \emph{dominating manipulation} for Proxy $j$ as a position $s_j'$ that satisfy the following conditions. First, by reporting $s_j'$, there exists a profile $p'\in \IS{s}{j^*}$ that when combined with $s_j'$, results in a more preferable outcome. Second, for all other profiles in $\IS{s}{j^*}$, it is the case that $j$ weakly prefers them over the current outcome.
	
	More formally, let $\succ_{j}$ be a full order over all possible outcomes that define $j$'s true preferences. 
	Then, $s_{j}'$ \emph{dominates} $s_j$ in state $(s,j^*)$ if for any profile $p \in \IS{s}{j^*}$ it holds  that $\wm{{s_{-j}},s_{j}', p} \succcurlyeq_j \wm{s, p} = s_{j^*}$, and the preference is strict for at least some $p'\in   \IS{s}{j^*}$.
	
	In the special case where $\sigma$ returns the true followers' positions (no uncertainty), dominating manipulations are just better-responses.
	
	\subsection{Convergence under Partial Information}
	
	In appendix~\ref{appx:worse_sc_incomplete_information}, we show that dominating-manipulations dynamics may converge to a worse social outcome than the truthful state. However, while this is possible for general policies, it is not the case for policies that can guarantee monotonicity. In this section, we find a sufficient condition for convergence to the true median for the partial information setting.
	
	We say that a policy for Proxy $j \in M$ is \emph{strong-monotone} if for every state $(s,j^*)$ and every profile $p \in \IS{s}{j^*}$, it holds that $s_j$ and $\pi_j(s)$ are on the same side of $\statemed{s,p}$.
	Note that for complete information, strong-monotone policies are reduced to monotone policies.

	\paragraph{The median interval.}
	For a dynamics $\dyn{s}$, let $s^t_r,s^t_{\ell}$ be the positions of the closest proxies to the winner from right and left respectively at $t$. 
	
	Let $I^t:=\{\statemed{s^t,p}: p\in \IS{s^t}{j^*(s^t)}\}$ be the interval of possible positions for the median in $s^t$, and denote $\ell^t,r^t$ the lower and upper bounds of $I^t$, respectively.
	
	We define an interval $I^t$ recursively as follows. For $t=0$, set $I^0=\parenthesis{\frac{s^0_r+s^0_{j^*}}{2},\frac{s^0_{\ell}+s^0_{j^*}}{2}}$. For every $t> 0$ define $I^t$ as follows:
	\begin{itemize}
		\item If $j^* \ne \pmov{t-1}$ then $I^t=I^{t-1}\cap \parenthesis{\frac{s^t_r+s^t_{\pwin{0}}}{2},\frac{s^t_{\ell}+s^t_{\pwin{0}}}{2}}$.
		\item  Otherwise, if $s^{t-1}_{\pwin{t}} < s^t_{\pwin{t}}$ set $I^t = I^{t-1} \cap \left[ s^t_{\pwin{t}}, \infty \right)$, else if $s^{t-1}_{\pwin{t}} > s^t_{\pwin{t}}$ set $I^t = I^{t-1} \cap \left( -\infty, s^t_{\pwin{t}} \right]$
	\end{itemize}

	\begin{remark}
		Note that the position $\swin{t}$ is known to proxies even though the profile $p$ is unknown, since $\swin{t}=s^t_{j^*}$ and both of $s^t$ and $j^*=j^*(s^t)$ are known. Hence the proxies can indeed infer $I^t$ from the information they know.
	\end{remark}
	
	\paragraph{Strong-monotone dominating manipulations for non-winning proxies.}
	Then, the set of dominating manipulations at state $(s^t,j^*)$ that are also strong-monotone for Proxy $j$ with $p_j \le\swin{t}$ is the open interval:
	$$\parenthesis{ \min\braces{s^t_{\ell},\swin{t} - 2\abs{\swin{t}-\ell^t}}, \swin{t} },$$
	if $\ell^t < \swin{t}$, and empty otherwise. The set for proxies on the other side of the median is similar with respect to $r^t$.
	
	\paragraph{Strong-monotone dominating manipulations for winning proxies.}
	For winning proxies, they  have a dominating manipulation only if their current position is between their peak and $I^t$, and the closest proxy on the other side of $I^t$ is farther than their position. That is, w.l.o.g assume that $p_{\pwin{t}} < \swin{t} < \ell^t$, and $\abs{r^t-s^t_r} > \abs{r^t-\swin{t}}$. In this case, their set is $\parenthesis{r^t-\abs{r^t-\swin{t}}, \swin{t}}$. This is the equivalent of a meta-move in this setting. Their set is empty in any other case since either any beneficial deviation may cross the median, or a deviating may have a negative outcome.
	
	We get that a policy is strong monotone if it holds that if the position of a proxy is left (right) of $\ell^t$ ($r^t$), then their resulting strategy would have the same orientation with respect to $I^t$. Furthermore, as $I^t$ decreases with every step of a moving proxy, and that in turn decreases the set of strong-monotone dominating manipulations for each proxy, these sets are monotonically decreasing. If we impose a similar restriction as in the proof of Corollary~\ref{corr:sufficient_step_implies_convergence}, we get convergence with the true median as an outcome with a similar argument.
	
	\subsection{Rationalizing Monotonicity}
	In decision theory, the concept of regret is often used to model types of agents. Given a strategic decision made under uncertainty, the regret is the difference in utilities between the outcome, and the optimal strategy that the agent could use ex-post. A regret-averse or risk-neutral agent would select the strategy that minimizes the maximal regret.
	
	In what follows, we show that the minimax-regret policy guarantees monotonicity, and therefore if all proxies are regret-averse then even with partial information convergence to the true median is guaranteed.
	
	The following theorem shows that the minimax-regret strategy is strong-monotone.
	
	\begin{theorem}
		\label{thm:minimax_regret}
		If the policy of every proxy is strong-monotone in a dynamics up to step $t\ge 0$, then the minimax-regret strategy of every proxy is strong-monotone.
	\end{theorem}
	\begin{proof}
		First, consider a proxy $j$ with peak $p_j$ left of the current winner.
		For a possible strategy $s_{j}'$, we calculate the maximal regret by distinguishing between the possible values of $s_{j}'$:
		\begin{itemize}
			\item $s_{j}' \le \ell^t-\abs{\ell^t-\swin{t}}$-- if the median is right of the current winner, then ex-post there is nothing Proxy $j$ can do to change the outcome in their favor, thus the regret is $0$. Otherwise, the optimal position for them is the position symmetric to the current winner with respect to the median, whereas by reporting $s_{j}'$ the outcome does not change. Thus, the difference in utility is the distance between the current winner and the optimal position. The maximal regret is reached when the median is at $\ell^t$ (up $\varepsilon>0$) and is equal to $2\cdot \abs{\ell^t-\swin{t}}$.
			\item $s_{j}' \ge \swin{t}$-- It is still possible that the median is (weakly) left of the current winner. By reporting a position that is right of the current winner the median would be bounded by the current winner and therefore the outcome will not change. Thus, the maximal regret is at least $2\cdot \abs{\ell^t-\swin{t}}$.
			\item $\ell^t- \abs{\ell^t-\swin{t}} < s_{j}' < \swin{t}$-- as in the first case, if the median is right of the current winner then the regret is $0$. Otherwise, for every possible position of $\med{t}$, the optimal strategy for $j$ ex-post is $opt^t=\med{t}-\abs{\med{t}-\swin{t}}$. Therefore, if $opt^t<s_{j}'$ then the regret is $s_{j}'-opt^t$. Otherwise, $s_{j}'$ is further than the current winner from the current median and thus the outcome will not change. Hence, the regret is $\swin{t}-opt^t$. We get that the maximal regret is $\max_{\med{t}}\braces{\swin{t}-opt^t,\parenthesis{s_{j}'-opt^t}\cdot\mathbbm{1}_{opt^t<s_{j}'}}$. 
		\end{itemize}
		
		\begin{figure}[ht!]
			\centering
			\begin{tikzpicture}[x=1cm, y=1cm, z=1mm]
				\coordinate (sl) at (1,0);
				\coordinate (l) at (2,0);
				\coordinate (st) at (3,0);
				\coordinate (r) at (5,0);
				\coordinate (sr) at (7,0);
				\coordinate (diffl) at (0,2);
				\coordinate (diffl2) at (0,1);
				\coordinate (diffr) at (0,4);
				
				\draw[-latex] (0,-0.1) -- (8,0);
				\draw (sl)++(0,0.1) -- ++(0,-0.2) node[below] {$s_\ell^t$};
				\draw (st)++(0,0.1) -- ++(0,-0.2) node[below] {$\swin{t}$};
				\draw (sr)++(0,0.1) -- ++(0,-0.2) node[below] {$s_r^t$};
				\draw[dashed] (l)++(0,5) -- ++(0,-5.1) node[below] {$\ell^t$};
				\draw[dashed] (r)++(0,5) -- ++(0,-5.1) node[below] {$r^t$};
				
				\draw[-latex] (0,-0.1) -- (0,5);
				\draw[dashed] (diffl)++(8,0) -- ++(-8.1,0) node[left] {$\abs{\swin{t}-s^t_\ell}$};
				\draw[dashed] (diffl2)++(8,0) -- ++(-8.1,0) node[left] {$\abs{\swin{t}-\ell^t}$};
				\draw[dashed] (diffr)++(8,0) -- ++(-8.1,0) node[left] {$\abs{\swin{t}-s^t_r}$};
				
				\draw[very thick] (diffl) -- ++(sl);
				\draw[very thick] (diffl-|sl) -- (diffl2-|l) -- (diffl-|st);
				\draw[very thick] (diffl)++(st) -- ++(2,0) -- (diffr-|sr);
				\draw[very thick] (diffr)++(sr) -- ++(1,0);
			\end{tikzpicture}
			\caption{Example of maximal regret values for a proxy with position left of the current winner.}
			\label{fig:my_label}
		\end{figure}
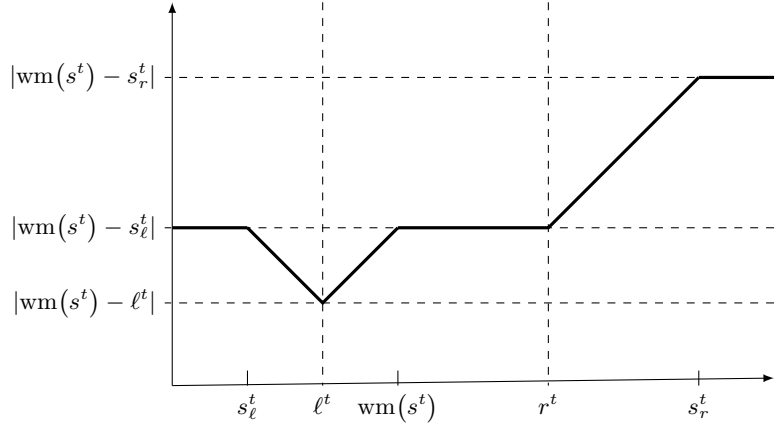
		
		If $\abs{\ell^t-\swin{t}} = 0$, then the regret for the first case is $0$, for the second it is positive and there are no strategies that fit the last case. Thus, the minimax-regret strategy in this case is the set $\parenthesis{-\infty,\ell^t}$, and every position in this set is strong-monotone.
		
		Otherwise, for the last case, $2\cdot \abs{\ell^t-\swin{t}}$ bounds the regret from above and for the other cases it bounds the regret from below, therefore the minimax regret must be attained for $\ell^t- \abs{\ell^t-\swin{t}} < s_{j}' < \swin{t}$. We argue that the minimax regret is attained at $\ell^t$. First, for $s_{j}'=\ell^t$, the maximal regret is $\abs{\ell^t-\swin{t}}$. To see this, if $opt_t>\ell^t$ then the regret is given by $\swin{t}-opt^t\le \abs{\ell^t-\swin{t}}$. Otherwise, the regret is given by
		$$\ell^t-opt^t\le \ell^t-\parenthesis{\ell^t- \abs{\ell^t-\swin{t}}}=\abs{\ell^t-\swin{t}}.$$
		Next, for $\ell^t < s_{j}' < \swin{t}$, if $\med{t}=\ell^t+\varepsilon$ for $\varepsilon>0$ then $$opt^t=\ell^t-\abs{\ell^t-\swin{t}}+2\varepsilon$$
		thus the regret is $s_{j}'-opt^t>\abs{\ell^t-\swin{t}}$. Finally, for $s_{j}' < \ell^t$ then if $\med{t} = \frac{\swin{t}+s_{j}'}{2}+\varepsilon$ then the regret is $\swin{t}-opt^t>\abs{\ell^t-\swin{t}}$.
		
		We get that the minimax-regret strategy for Proxy $j$ is $\ell^t \le \med{t}$. Thus, it is strong-monotone.
		
		For proxies with positions right of the current winner, the analysis is symmetric for $r^t$.
		
		For the winning proxy, if their reported position is at their peak, then it is the optimal position for them regardless of the underlying state. Thus, it is also their minimax-regret strategy (the maximal regret of their peak is $0$).
		
		Otherwise, w.l.o.g assume that the winning proxy's truthful position (peak) is left of their current position. By definition of $I^t$ and since all steps until $t$ are strong-monotone, it follows that $\ell^t \ge \swin{t}$.
		
		Consider the current position of the winning proxy. Their maximal regret is attained for the case that the median is at their current position, in this case their optimal position is at distance $\min\braces{\abs{s^t_r-\swin{t}},\abs{s^t_{\ell}-\swin{t}}}$ left of their position (up to $\varepsilon$). Thus, their maximal regret at this position is bounded from below by $\min\braces{\abs{s^t_r-\swin{t}},\abs{s^t_{\ell}-\swin{t}}}$. For every position right of their current position, for the same case the regret would be at least $\min\braces{\abs{s^t_r-\swin{t}},\abs{s^t_{\ell}-\swin{t}}}$, so it bounds the maximal regret from below. Finally, for every position left of their current position, if the median is at $r^t-\varepsilon$ then the winning position would be $s^t_r$, so the regret is bounded from below by $\abs{s^t_r-\swin{t}}\ge \min\braces{\abs{s^t_r-\swin{t}},\abs{s^t_{\ell}-\swin{t}}}$. Thus, their minimax-regret strategy is their current position, which is strong-monotone.\qed
	\end{proof}
	
	\begin{remark}
		Note that the minimax-regret policy is also in the set of dominating manipulations.
	\end{remark}

	Consequently, it follows that if the policy of all proxies is minimax-regret, then the resulting dynamics is strong-monotone and therefore converges to the true median.

	
	\section{Conclusions and Future Work}
	
	We introduced \emph{Strategic Proxy Games}, a framework to study strategic behavior of proxies in voting mechanisms.
	
	First, we demonstrated that in this model, the extension of the median voting rule to the weighted median voting rule via proxy voting maintains strategyproofness with respect to followers' positions. In particular, this suggests that with respect to follower positions, the delegation scheme is optimal for followers preferences. Our study uses the Tullock delegation scheme, however, other delegation models have been studied in the literature. In the one-step delegation domain, ~\citet{green2015direct} consider delegation that accounts for small errors in assessment of positions, and ~\citet{alon2015robust} consider social connections that influence the weight of proxies. Exploring the impact of different delegation models on the outcome of proxy voting and the strategic behavior of followers and proxies would be an interesting direction for future research. We point out that many of our results depend on the fact that the proxy who attracts the median voter wins. Our conjecture is that for delegation models that correlate well with distance the same property hold, and would yield similar results.
	In this research we focused on the median voting rule. We plan to study the implication of strategic proxy behavior in higher dimensions, as well as with other voting rules.
	
	We continued to study the strategic behavior of proxies, and showed that while strategyproofness does not extend to proxy voting, the distance of the outcome in a setting of repeated manipulations is bounded by the distance of the truthful outcome. Thus, in terms of distance, manipulations can only have a positive impact on the outcome. Moreover, when proxies maintain the integrity of their positions with respect to the median, the outcome converges to the social optimum. We further show that for discrete spaces, non-monotonicity can result in a worse social outcome. The combination of the above results show that in the context of proxy voting, both complete truthfulness and unbounded manipulation are sub-optimal. The assumption of monotonicity is a compromise between these extremities.
	In future work, we plan to further study non-monotone settings. In particular, we showed that truth-orientation bounds the distance of the outcome for the social optimum. We conjecture that under the stronger assumption of truth-bias, proxies would have an incentive to revert to a monotone state.
	
	Finally, we study the implications of partial information to the strategic behavior of proxies. While our results show that the outcome may increase the social cost, we also show that for policies that guarantee to be monotone, in particular minimax-regret, would converge to the social optimum. While the public information for voters that we consider is very minimal, i.e. it only includes the winning position and the positions of proxies, it may be possible to achieve the same positive results with even less information available. A particular case of interest is when the winning position is not made public, rather an estimate of it. This case may be more realistic for several reasons. First, if voters positions are estimated, then the outcome can be an estimate as well. This may be appropriate for settings where followers, for reason of e.g. cognitive strain or privacy, wish to communicate an interval of approved positions that expresses an estimate of their peak. This setting is somewhat related to that of ~\citet{green2015direct} mentioned above, and the setting proposed by~\citet{feldman2016variations}, where candidates admit attraction intervals, and voters may select any candidate that attracts them. Moreover, approximate positions of voters are more realistic in the context of polls.
	Another possibility is to allow for more strategic depth for proxies, where they communicate an approximate positions in an effort to maximize support.

	
	\subsubsection*{Acknowledgements.}
	This research was supported by the Israel Science Foundation (ISF; Grant No. 2539/20).

	\newpage
	
	\bibliography{mybib}

\begin{thebibliography}{34}
\providecommand{\natexlab}[1]{#1}
\providecommand{\url}[1]{\texttt{#1}}
\expandafter\ifx\csname urlstyle\endcsname\relax
  \providecommand{\doi}[1]{doi: #1}\else
  \providecommand{\doi}{doi: \begingroup \urlstyle{rm}\Url}\fi

\bibitem[Alon et~al.(2015)Alon, Feldman, Lev, and Tennenholtz]{alon2015robust}
N.~Alon, M.~Feldman, O.~Lev, and M.~Tennenholtz.
\newblock How robust is the wisdom of the crowds?
\newblock In \emph{Twenty-Fourth International Joint Conference on Artificial
  Intelligence}. Citeseer, 2015.

\bibitem[Bielous and Meir(2022)]{bielous2022proxy}
G.~Bielous and R.~Meir.
\newblock Proxy manipulation for better outcomes.
\newblock In \emph{Multi-Agent Systems: 19th European Conference, EUMAS 2022,
  D{\"u}sseldorf, Germany, September 14--16, 2022, Proceedings}, pages 79--95.
  Springer, 2022.

\bibitem[Black(1948)]{black1948rationale}
D.~Black.
\newblock On the rationale of group decision-making.
\newblock \emph{Journal of political economy}, 56\penalty0 (1):\penalty0
  23--34, 1948.

\bibitem[Br{\^a}nzei et~al.(2013)Br{\^a}nzei, Caragiannis, Morgenstern, and
  Procaccia]{branzei2013bad}
S.~Br{\^a}nzei, I.~Caragiannis, J.~Morgenstern, and A.~Procaccia.
\newblock How bad is selfish voting?
\newblock In \emph{Proceedings of the AAAI conference on artificial
  intelligence}, volume~27, 2013.

\bibitem[Caragiannis and Micha(2019)]{caragiannis2019contribution}
I.~Caragiannis and E.~Micha.
\newblock A contribution to the critique of liquid democracy.
\newblock In \emph{IJCAI}, pages 116--122, 2019.

\bibitem[Cohensius et~al.(2017)Cohensius, Manor, Meir, Meirom, and
  Orda]{cohensius2016proxy}
G.~Cohensius, S.~Manor, R.~Meir, E.~Meirom, and A.~Orda.
\newblock Proxy voting for better outcomes.
\newblock In \emph{AAMAS'17}, 2017.

\bibitem[Conitzer et~al.(2011)Conitzer, Walsh, and Xia]{conitzer2011dominating}
V.~Conitzer, T.~Walsh, and L.~Xia.
\newblock Dominating manipulations in voting with partial information.
\newblock In \emph{Twenty-Fifth AAAI Conference on Artificial Intelligence},
  2011.

\bibitem[Desmedt and Elkind(2010)]{desmedt2010equilibria}
Y.~Desmedt and E.~Elkind.
\newblock Equilibria of plurality voting with abstentions.
\newblock In \emph{Proceedings of the 11th ACM conference on Electronic
  commerce}, pages 347--356, 2010.

\bibitem[Downs(1957)]{downs1957economic}
A.~Downs.
\newblock \emph{An economic theory of democracy}.
\newblock Harper \& Row New York, 1957.

\bibitem[Dutta et~al.(2001)Dutta, Jackson, and Le~Breton]{dutta2001strategic}
B.~Dutta, M.~O. Jackson, and M.~Le~Breton.
\newblock Strategic candidacy and voting procedures.
\newblock \emph{Econometrica}, 69\penalty0 (4):\penalty0 1013--1037, 2001.

\bibitem[Feldman et~al.(2016)Feldman, Fiat, and
  Obraztsova]{feldman2016variations}
M.~Feldman, A.~Fiat, and S.~Obraztsova.
\newblock Variations on the hotelling-downs model.
\newblock In \emph{Proceedings of the AAAI Conference on Artificial
  Intelligence}, volume~30, 2016.

\bibitem[Gehrlein and Lepelley(2010)]{gehrlein2010voting}
W.~V. Gehrlein and D.~Lepelley.
\newblock \emph{Voting paradoxes and group coherence: the Condorcet efficiency
  of voting rules}.
\newblock Springer Science \& Business Media, 2010.

\bibitem[Ghodsi et~al.(2019)Ghodsi, Latifian, and
  Seddighin]{ghodsi2019distortion}
M.~Ghodsi, M.~Latifian, and M.~Seddighin.
\newblock On the distortion value of the elections with abstention.
\newblock In \emph{Proceedings of the AAAI Conference on Artificial
  Intelligence}, volume~33, pages 1981--1988, 2019.

\bibitem[Gibbard(1973)]{gibbard1973manipulation}
A.~Gibbard.
\newblock Manipulation of voting schemes: a general result.
\newblock \emph{Econometrica: journal of the Econometric Society}, pages
  587--601, 1973.

\bibitem[G{\"o}lz et~al.(2021)G{\"o}lz, Kahng, Mackenzie, and
  Procaccia]{golz2021fluid}
P.~G{\"o}lz, A.~Kahng, S.~Mackenzie, and A.~D. Procaccia.
\newblock The fluid mechanics of liquid democracy.
\newblock \emph{ACM Transactions on Economics and Computation}, 9\penalty0
  (4):\penalty0 1--39, 2021.

\bibitem[Grandi et~al.(2013)Grandi, Loreggia, Rossi, Venable, and
  Walsh]{grandi2013restricted}
U.~Grandi, A.~Loreggia, F.~Rossi, K.~B. Venable, and T.~Walsh.
\newblock Restricted manipulation in iterative voting: Condorcet efficiency and
  borda score.
\newblock In \emph{Algorithmic Decision Theory: Third International Conference,
  ADT 2013, Bruxelles, Belgium, November 12-14, 2013, Proceedings 3}, pages
  181--192. Springer, 2013.

\bibitem[Green-Armytage(2015)]{green2015direct}
J.~Green-Armytage.
\newblock Direct voting and proxy voting.
\newblock \emph{Constitutional Political Economy}, 26\penalty0 (2):\penalty0
  190--220, 2015.

\bibitem[Halpern et~al.(2021)Halpern, Halpern, Jadbabaie, Mossel, Procaccia,
  and Revel]{halpern2021defense}
D.~Halpern, J.~Y. Halpern, A.~Jadbabaie, E.~Mossel, A.~D. Procaccia, and
  M.~Revel.
\newblock In defense of liquid democracy.
\newblock \emph{arXiv preprint arXiv:2107.11868}, 2021.

\bibitem[Hotelling(1929)]{hotelling1929stability}
H.~Hotelling.
\newblock Stability in competition.
\newblock \emph{The Economic Journal}, 39\penalty0 (153):\penalty0 41--57,
  1929.

\bibitem[J{\"o}nsson and {\"O}rnebring(2011)]{jonsson2011user}
A.~M. J{\"o}nsson and H.~{\"O}rnebring.
\newblock User-generated content and the news: Empowerment of citizens or
  interactive illusion?
\newblock \emph{Journalism Practice}, 5\penalty0 (2):\penalty0 127--144, 2011.

\bibitem[Kahng et~al.(2021)Kahng, Mackenzie, and Procaccia]{kahng2021liquid}
A.~Kahng, S.~Mackenzie, and A.~Procaccia.
\newblock Liquid democracy: An algorithmic perspective.
\newblock \emph{Journal of Artificial Intelligence Research}, 70:\penalty0
  1223--1252, 2021.

\bibitem[Meir(2017)]{meir2017iterative}
R.~Meir.
\newblock Iterative voting.
\newblock \emph{Trends in computational social choice}, pages 69--86, 2017.

\bibitem[Meir(2018)]{meir2018strategic}
R.~Meir.
\newblock Strategic voting.
\newblock \emph{Synthesis lectures on artificial intelligence and machine
  learning}, 13\penalty0 (1):\penalty0 1--167, 2018.

\bibitem[Meir et~al.(2010)Meir, Polukarov, Rosenschein, and
  Jennings]{meir2010convergence}
R.~Meir, M.~Polukarov, J.~Rosenschein, and N.~Jennings.
\newblock Convergence to equilibria in plurality voting.
\newblock In \emph{Proceedings of the AAAI conference on artificial
  intelligence}, volume~24, pages 823--828, 2010.

\bibitem[Meir et~al.(2014)Meir, Lev, and Rosenschein]{meir2014local}
R.~Meir, O.~Lev, and J.~S. Rosenschein.
\newblock A local-dominance theory of voting equilibria.
\newblock In \emph{Proceedings of the fifteenth ACM conference on Economics and
  computation}, pages 313--330, 2014.

\bibitem[Moulin(1980)]{moulin1980strategy}
H.~Moulin.
\newblock On strategy-proofness and single peakedness.
\newblock \emph{Public Choice}, 35\penalty0 (4):\penalty0 437--455, 1980.

\bibitem[Petrik(2009)]{petrik2009participation}
K.~Petrik.
\newblock Participation and e-democracy how to utilize web 2.0 for policy
  decision-making.
\newblock In \emph{10th Annual International Conference on Digital Government
  Research: Social Networks: Making Connections between Citizens, Data and
  Government}, pages 254--263, 2009.

\bibitem[Procaccia and Rosenschein(2006)]{procaccia2006distortion}
A.~D. Procaccia and J.~S. Rosenschein.
\newblock The distortion of cardinal preferences in voting.
\newblock In \emph{International Workshop on Cooperative Information Agents},
  pages 317--331. Springer, 2006.

\bibitem[Reijngoud and Endriss(2012)]{reijngoud2012voter}
A.~Reijngoud and U.~Endriss.
\newblock Voter response to iterated poll information.
\newblock In \emph{Proceedings of the 11th International Conference on
  Autonomous Agents and Multiagent Systems-Volume 2}, pages 635--644, 2012.

\bibitem[Riddick and Butcher(1991)]{riddick1991riddick}
F.~M. Riddick and M.~H. Butcher.
\newblock \emph{Riddick's Rules of procedure: a modern guide to faster and more
  efficient meetings}.
\newblock Madison Books, 1991.

\bibitem[Sabato et~al.(2017)Sabato, Obraztsova, Rabinovich, and
  Rosenschein]{sabato2017real}
I.~Sabato, S.~Obraztsova, Z.~Rabinovich, and J.~S. Rosenschein.
\newblock Real candidacy games: a new model for strategic candidacy.
\newblock In \emph{Proceedings of the 16th Conference on Autonomous Agents and
  MultiAgent Systems}, pages 867--875, 2017.

\bibitem[Satterthwaite(1975)]{satterthwaite1975strategy}
M.~A. Satterthwaite.
\newblock Strategy-proofness and arrow's conditions: Existence and
  correspondence theorems for voting procedures and social welfare functions.
\newblock \emph{Journal of economic theory}, 10\penalty0 (2):\penalty0
  187--217, 1975.

\bibitem[Schaupp and Carter(2005)]{schaupp2005voting}
L.~C. Schaupp and L.~Carter.
\newblock E-voting: from apathy to adoption.
\newblock \emph{Journal of Enterprise Information Management}, 2005.

\bibitem[Tullock(1967)]{tullock1967proportional}
G.~Tullock.
\newblock Proportional representation.
\newblock \emph{Toward a mathematics of politics}, pages 144--157, 1967.

\end{thebibliography}
	
	\newpage
	
	\appendix
	
	\section{Convergence to sub-optimal outcome for discrete spaces}
	\label{appx:worse_sc_discrete}
	
	Consider the following SPG. The set of proxies is $M=\braces{1,2,3,4,5}$ with positions $p_{4}=p_{3}=p_{2}=-13$, $p_{1}=-11$ and $p_{5}=12$. In addition, there are four followers $N=\braces{6,7,8,9}$ such that $p_6=p_7=5$, $p_8=1$ and $p_9=0$ is the true median. For the truthful state the winner is Proxy $1$ and the outcome is $-11$. The social cost is
	\begin{align*}
		SC &= 3\cdot\abs{-13 - \parenthesis{-11}}+\abs{0-\parenthesis{-11}}+\abs{1-\parenthesis{-11}}+2\cdot\abs{5-\parenthesis{-11}}+\abs{12-\parenthesis{-11}} \\
		&= 3\cdot 2 + 11 + 12 + 2\cdot 16 + 23 = 84
	\end{align*}
	
	Next, Proxy $5$ is the only proxy that has a better-response, and they report the position $\smov{0}=10$ and wins. At the next steps, Proxies $2,3$ and $4$ are the moving one-by-one, reporting the positions $9,8$ and $7$ respectively. At each step the moving proxy also wins, and it can be easily verifies that these positions are better responses. Then, Proxy $1$ reports $\smov{5}=4$, and wins. Consider the state $\parenthesis{s^5_{-5}, 6}$. Recall that we assume the existence of a deterministic tie-breaking scheme. We distinguish between the following:
	\begin{itemize}
		\item Proxy $1$ wins at position $4$. In this case, Proxy $5$ reports the position $5$, that is $s^6=\parenthesis{s^5_{-5}, 6}$. We argue that this a PNE. For Proxy $5$, since the distance of Proxy $1$ from the current median $5$ in $s^6$ is $1$, and $1$ would win the tie-break if Proxy $5$ reports $6$, then their better-response set is empty. For the other proxies, since Proxy $5$'s reported position is at the current median, they cannot get closer to the median. Moreover, no single proxy can change the position of the current median. Therefore, they do not have a better response.
		\item Proxy $5$ wins at position $6$. In this case $5$ reports the position $6$, that is, $s^6=\parenthesis{s^5_{-5},6}$. Then Proxy $1$ reports $5$, i.e. $s^7=\parenthesis{s^6_{-1},5}$. We argue that this is a PNE, with a similar argument.
	\end{itemize}
	
	Figure~\ref{fig:discrete_converge_worse} demonstrates the dynamics for the first case.
	
	\begin{figure}[ht!]
		\centering
		\begin{tikzpicture}[
			line/.style={draw=black, very thick}, 
			x=4.5mm, y=8mm, z=1mm]
			\coordinate (p1) at (-13,0);
			\coordinate (p2) at (-13,0);
			\coordinate (p3) at (-13,0);
			\coordinate (p4) at (-11,0);
			\coordinate (p5) at (12,0);
			\coordinate (f1) at (0,0);
			\coordinate (f2) at (1,0);
			\coordinate (f3) at (5,0);
			\coordinate (f4) at (5,0);
			
			\draw[line, -latex] (-14,0)--(13,0);
			\draw (p1)+(0,0.2) -- ++(0,-0.2) node[below] {$-13$};
			\draw (p1)+(0,0.5) node[circle,draw=black,fill=black!30!white,inner sep=3pt] {};
			\draw (p2)+(0,1) node[circle,draw=black,fill=black!30!white,inner sep=3pt] {};
			\draw (p3)+(0,1.5) node[circle,draw=black,fill=black!30!white,inner sep=3pt] {};
			
			\draw (f4)+(2,0.5) node[circle,draw=black,fill=black,inner sep=3pt] {};
			\draw (f4)+(3,0.5) node[circle,draw=black,fill=black,inner sep=3pt] {};
			\draw (f4)+(4,0.5) node[circle,draw=black,fill=black,inner sep=3pt] {};
			
			\draw (p4)+(0,0.2) -- ++(0,-0.2) node[below] {$-11$};
			\draw (p4)+(0,0.5) node[circle,draw=black,fill=black!30!white,inner sep=3pt] {};
			\draw (f4)+(-1,0.5) node[circle,draw=black,fill=black,inner sep=3pt] {};
			
			\draw (f1)+(0,0.2) -- ++(0,-0.2) node[below] {$0$};
			\draw (f1)+(0,0.5) node[circle,draw=black,fill=black,inner sep=1pt] {};
			
			\draw (f2)+(0,0.2) -- ++(0,-0.2) node[below] {$1$};
			\draw (f2)+(0,0.5) node[circle,draw=black,fill=black,inner sep=1pt] {};
			
			\draw (f3)+(0,0.2) -- ++(0,-0.2) node[below] {$5$};
			\draw (f3)+(0,0.5) node[circle,draw=black,fill=black,inner sep=1pt] {};
			\draw (f4)+(0,1) node[circle,draw=black,fill=black,inner sep=1pt] {};
			
			\draw (p5)+(0,0.2) -- ++(0,-0.2) node[below] {$12$};
			\draw (p5)+(0,0.5) node[circle,draw=black,fill=black!30!white,inner sep=3pt] {};
			\draw (p5)+(-2,0.5) node[circle,draw=black,fill=white,inner sep=3pt] {};
			\draw (f4)+(0,1.5) node[circle,draw=black,fill=black,inner sep=3pt] {};
			
			\draw[-latex,dashed] (p5)+(0,0.7) to[in=45,out=135] ++(-2,0.7) to[in=45,out=135] ++(-5,1.2);
			\draw[-latex,dashed] (p4)+(0,0.7) to[in=135,out=45] ++(15,0.7);
		\end{tikzpicture}
		\caption{The SPG, large dots indicate proxies, small dots are followers. Black are final positions, white are intermediate positions and gray are original positions.}
		\label{fig:discrete_converge_worse}
	\end{figure}
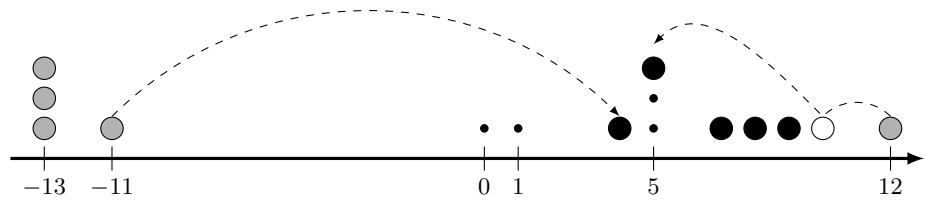
	
	In either case we get a PNE where the outcome is $5$. The social cost is
	\begin{align*}
		SC' &= 3\cdot\abs{-13 - 5}+\abs{-11-5}+\abs{0-5}+\abs{1-5}+\abs{12-5} \\
		&= 3\cdot 18 + 16 + 5 + 4 + 7 = 86
	\end{align*}
	That is strictly greater than the social cost of the truthful outcome.
	
	Note that for the unrestricted (non-discrete) setting, this is not a PNE since the winner in either case has a better-response, for example by reporting a position that is at a distance of $0.5$ from the median in the direction of their peak.
	
	\section{Convergence to sub-optimal outcome for incomplete information}
	\label{appx:worse_sc_incomplete_information}
	
	Consider the SPG appearing in Figure~\ref{fig:partial_information_unrestricted_simple}.
	
	
	\begin{figure}[ht!]
		\centering
		\begin{tikzpicture}[
			line/.style={draw=black, very thick}, 
			darkstyle/.style={circle,draw,fill=gray!40,minimum size=20},
			x=0.8mm, y=0.8mm, z=1mm]
			
			\draw[line, -](-50,30)--(90,30);
			
			\foreach \x/\t in {-50/0,-30/1,0/0,10/0,90/1} {
				\draw[line](\x,28)--(\x,32) node[below=10]{$\x$};
				\ifthenelse{\t=1}
				{\draw (\x,35) node[circle,draw=black,fill,inner sep=\proxyradius pt] {};}
				{\draw (\x,35) node[circle,draw=black,fill,inner sep=\voterradius pt] {};}
			}
		\end{tikzpicture}
		\caption{The SPG, large dots are proxies, small dots are followers.}
		\label{fig:partial_information_unrestricted_simple}
	\end{figure}
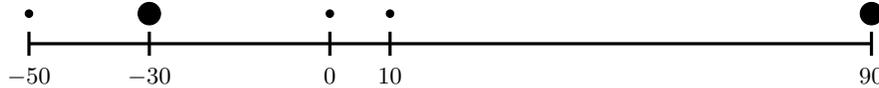
	
	The true positions are $p = \parenthesis{-30,90,-50,0,10}$. There are two proxies $M = \braces{1,2}$ with positions $p_{1} = -30$ and $p_{2} = 90$, and $3$ followers. Note that proxies and followers are unaware to the positions of other followers. The median is $0$, and the weighted median is $-30$. The social cost is:
	\begin{align*}
		SC = \abs{-50 &- \parenthesis{-30}} + \abs{-30 - \parenthesis{-30}} \\ &+ \abs{-30 - 0} + \abs{-30 - 10} + \abs{-30 -90}\\ &= 20 + 0 + 30 + 40 + 120 = 210
	\end{align*}
	
	The proxies can deduce from their available information that the true median is in the interval $\parenthesis{-\infty,30}$.
	
	Next, for Proxy $2$, consider the position ${s^1_{2}}' = 29$. Note that it is a dominating manipulation. The outcome of $s^2=\parenthesis{s^1_{-2},{s^1_{2}}'}$ is ${s^1_{2}}'=29$. It follows that the true median is in the interval $\parenthesis{-0.5,30}$. In particular, note that the new information that the proxies get has no effect on the right bound of the interval.
	
	Next, for Proxy $1$ in $s^2$, the only information that Proxy $1$ has is that their position is $-30$, and that the position of Proxy $2$ in $s^2$ is $29$, and it is the outcome of $s^2$. Consider the position ${s^2_{1}}' = 25$. Again, this is a dominating manipulation for Proxy $1$. The outcome of $s^3=\parenthesis{s^2_{-1},{s^2_{1}}'}$ is ${s^2_{1}}'=25$. Therefore, the true median is in the interval $\parenthesis{-0.5,27}$.
	
	Consider $s^3$. For Proxy $1$, since the true median can be at $27-\varepsilon$ for every $\varepsilon>0$, every position left of their current position may prove to be further than Proxy $2$'s position in $s^3$, resulting in a worse outcome. Thus, their set of dominating manipulations is empty. The set of dominating manipulations for Proxy $2$ is $\parenthesis{25,29}$. For every position ${s^3_{2}}' \in \parenthesis{25,29}$, the outcome of $\parenthesis{s^3_{-2}, {s^3_{2}}'}$ is $s^3_{2} = 25$. The true median is in the interval $\parenthesis{-0.5,\frac{25+{s^3_{2}}'}{2}}$ and the set of dominating manipulations for Proxy $2$ is $\parenthesis{25,{s^3_{2}}'}$. Note that this holds for every $t\ge 3$ s.t. for every $3 \le t' \le t$ the position of Proxy $1$ in $s^{t'}$ is $25$. That is, the outcome is $25$, the true median is in the interval $\parenthesis{25,\frac{25+{s^t_{2}}'}{2}}$, the set of dominating manipulations for Proxy $1$ is empty, and the set of dominating manipulations for Proxy $2$ is $\parenthesis{25, {s^t_{2}}'}$. Thus, Proxy $1$'s position does not change, and the sequence of positions of Proxy $2$ is monotonically decreasing and bounded, therefore it has a limit. We get that the dynamics converges, and the outcome of the limit state is $25$. Figure ~\ref{fig:partial_information_unrestricted} demonstrates the dynamics.
	
	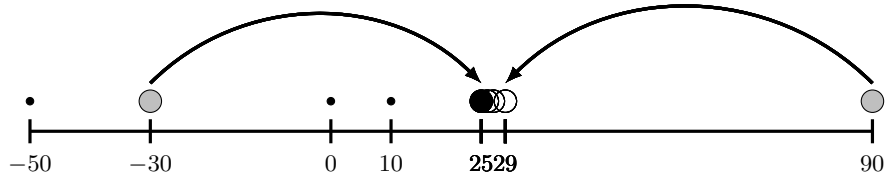
\begin{figure}[ht!]
		\centering
		\begin{tikzpicture}[
			line/.style={draw=black, very thick},
			proxy/.style={circle,draw=black,fill,inner sep=\proxyradius pt},
			emptyproxy/.style={circle,draw=black,fill=none,inner sep=\proxyradius pt},
			voter/.style={circle,draw=black,fill,inner sep=\voterradius pt},
			x=0.8mm, y=0.8mm, z=1mm]
			
			\draw[line, -](-50,0)--(90,0);
			
			\foreach \x/\t in {-50/0,-30/1,0/0,10/0,90/1} {
				\draw[line](\x,-2)--(\x,2) node[below=10]{$\x$};
				\ifthenelse{\t=1}
				{\draw (\x,5) node[circle,draw=black,fill=gray!50,inner sep=\proxyradius pt] {};}
				{\draw (\x,5) node[voter] {};}
				
				\foreach \x in {29,27,26} {
					\draw (\x,5) node[emptyproxy] {};
				}
				\draw[line](29,-2)--(29,2) node[below=10]{$29$};
				\draw (25,5) node[proxy] {};
				\draw[line](25,-2)--(25,2) node[below=10]{$25$};
				
				\draw[line,-latex] (90,8) to[out=135,in=45] (29,8);
				\draw[line,-latex] (-30,8) to[out=45,in=135] (25,8);
				
			}
		\end{tikzpicture}
		\caption{An example that converges to a worse outcome than the truthful state. Gray dots indicate truthful positions of proxies, empty dots indicate positions of dominating manipulations, full dot indicate limit positions. Small full dots are followers.}
		\label{fig:partial_information_unrestricted}
	\end{figure}
	
	The social cost of the outcome is
	\begin{align*}
		SC' = \abs{-50 &- 25} + \abs{-30 - 25} \\ &+ \abs{0 - 25} + \abs{10 - 25} + \abs{90 -25}\\ &= 75 + 55 + 25 + 15 + 65 = 235
	\end{align*}
	
	Note that complete information prevents the convergence demonstrated by this counter-example. With complete information, the true median is public information. Proxy $1$ can exploit this to strategically report a symmetric position on the opposite side of their peak. Thus, their better-response set is not empty.
	
\end{document}